\documentclass[a4paper,USenglish]{lipics-v2021}
\hideLIPIcs \nolinenumbers

\usepackage{amstext,amsgen,latexsym,amsmath}
\usepackage{amssymb,amsfonts}
\usepackage{amsthm} 

\usepackage{eclbkbox}

\usepackage{framed}
\usepackage{tcolorbox}
\tcbuselibrary{breakable}
\usepackage{algpseudocode,algorithm}

\usepackage{pifont}
\usepackage{hyperref}
\usepackage{mdframed}
\usepackage{url}

\newtheorem{openproblem}{Open Problem}
\newtheorem{problem}{Problem}

\newcommand{\ignore}[1]{}
\newcommand{\abs}[1]{| #1 |}

\title{Distributed Complexity of $P_k$-freeness:\\ Decision and Certification} 

\titlerunning{Distributed Complexity of $P_k$-freeness: Decision and Certification} 

\author{Masayuki Miyamoto}{Graduate School of Mathematics, Nagoya University}{masayuki.miyamoto95@gmail.com}{}{}

\authorrunning{M. Miyamoto} 

\Copyright{Masayuki Miyamoto} 



\category{} 

\acknowledgements{The author would like to thank Fran\c{c}ois Le Gall for discussions and for commenting on an earlier version of this paper. The author also acknowledges the anonymous reviewers for their fruitful comments.}

\EventEditors{John Q. Open and Joan R. Access}
\EventNoEds{2}
\EventLongTitle{42nd Conference on Very Important Topics (CVIT 2016)}
\EventShortTitle{CVIT 2016}
\EventAcronym{CVIT}
\EventYear{2016}
\EventDate{December 24--27, 2016}
\EventLocation{Little Whinging, United Kingdom}
\EventLogo{}
\SeriesVolume{42}
\ArticleNo{23}
\begin{document}
\maketitle
\begin{abstract}
The class of graphs that do not contain a path on $k$ nodes as an induced subgraph ($P_k$-free graphs) has rich applications in the theory of graph algorithms. This paper explores the problem of deciding $P_k$-freeness from the viewpoint of distributed computing. 

For specific small values of $k$, we present the \textit{first} $\mathsf{CONGEST}$ algorithms specified for $P_k$-freeness, utilizing structural properties of $P_k$-free graphs in a novel way. Specifically, we show that $P_k$-freeness can be decided in $\tilde{O}(1)$ rounds for $k=4$ in the $\mathsf{broadcast\;CONGEST}$ model, and in $\tilde{O}(n)$ rounds for $k=5$ in the $\mathsf{CONGEST}$ model, where $n$ is the number of nodes in the network and $\tilde{O}(\cdot)$ hides a $\mathrm{polylog}(n)$ factor. These results significantly improve the previous $O(n^{2-2/(3k+2)})$ upper bounds by Eden et al. (Dist.~Comp.~2022). The main technical contribution is a novel technique used in our algorithm for $P_5$-freeness to distinguish induced $5$-paths from non-induced ones, which is potentially applicable to other induced subgraphs. This technique also enables the construction of a local certification of $P_5$-freeness with certificates of size $\tilde{O}(n)$. This is nearly optimal, given our $\Omega(n^{1-o(1)})$ lower bound on certificate size, and marks a significant advancement as no nontrivial bounds for proof-labeling schemes of $P_5$-freeness were previously known. 

For general $k$, we establish the first $\mathsf{CONGEST}$ lower bound, which is of the form $n^{2-1/\Theta(k)}$. The $n^{1/\Theta(k)}$ factor is unavoidable, in view of the $O(n^{2-2/(3k+2)})$ upper bound mentioned above. Additionally, our approach yields the \textit{first} superlinear lower bound on certificate size for local certification. This partially answers the conjecture on the optimal certificate size of $P_k$-freeness, asked by Bousquet et al. (arXiv:2402.12148). 

Finally, we propose a novel variant of the problem called ordered $P_k$ detection. We show that in the $\mathsf{CONGEST}$ model, the round complexity of ordered $P_k$ detection is $\tilde{\Omega}(n)$ for $k \geq 5$, and in contrast, proving any nontrivial lower bound for ordered $P_3$ detection implies a strong circuit lower bound. As a byproduct, we establish a circuit-complexity barrier for $\Omega(n^{1/2+\varepsilon})$ \textit{quantum} $\mathsf{CONGEST}$ lower bounds for induced $4$-cycle detection. This is complemented by our $\tilde{O}(n^{3/4})$ quantum upper bound, which surpasses the classical $\tilde{\Omega}(n)$ lower bound by Le Gall and Miyamoto (ISAAC~2021).
\end{abstract}

\maketitle

\section{Introduction}
\subsection{Background}
The subgraph detection problems in the distributed setting of limited communication bandwidth, including the $\mathsf{CONGEST}$ model, are interesting and well studied problems. 
For a given specific small pattern graph $H$ (the description of $H$ is known to all nodes of the network), the goal of $H$ detection is to decide if the network $G$ contains $H$ as a subgraph (or an induced subgraph). If the network $G$ contains $H$, at least one node outputs ``yes'', otherwise all nodes of the network output ``no''. ($H$ detection is often referred as $H$ freeness, where at least one node outputs ``no'' if the network contains $H$, and otherwise all nodes of the network output ``yes''. In this paper we abuse these two problems since they are essentially equivalent.) There are several variants previously considered in the literature: $H$ listing requires to output all copies of $H$ in $G$ (each node outputs some copies of $H$ and the union of all output lists are equivalent to the list of all copies of $H$ in $G$). 
In recent years, significant progress in distributed subgraph detection has been made, and we now understand well the round complexities for detection of e.g., cliques and cycles. However, much less is known about induced paths. We summarize the previous results in this context in Section~\ref{appendix:related work}.

\paragraph*{$P_k$-freeness in the distributed setting}
Since it is known that in the $\mathsf{CONGEST}$ model, detecting non-induced $k$-paths can be done in $O(1)$ rounds for constant $k$~\cite{korhonen_et_al:LIPIcs:2018:8625}, our focus in this paper is specifically on the induced $k$-path detection problem. Induced $k$-paths are particularly more important than non-induced paths because graphs that do not contain an induced $k$-path (often called $P_k$-free graphs) have numerous applications in algorithmic graph theory, particularly in the centralized setting. The absence of induced $k$-node paths imposes structural constraints on graphs, which can be leveraged to develop efficient algorithms for problems like independent set or coloring~\cite{brause20154,gartland2020independent,grzesik2022polynomial,lokshantov2014independent,hoang2010deciding,randerath20043}. Despite its importance, few results are known for $P_k$-freeness in the distributed setting. The only known result we are aware of is the result of Nikabadi and Korhonen~\cite{nikabadi2022beyond}, who explored the multicolored variant of the induced $k$-path detection, where each node is colored by an integer from $\{1,2,\ldots,k\}$, and the goal is to detect an induced path on $k$ different colors. They demonstrated that multicolored induced $P_k$ detection requires $\Omega(n/\log n)$ rounds for any $k\geq 7$. However, their proof fails for $P_k$-freeness. To the best of our knowledge, there is no nontrivial result\footnote{For the trivial case of $k=3$, $P_3$ detection can be done in $O(1)$ rounds} for $P_k$-freeness in the $\mathsf{CONGEST}$ model beyond the general result by~\cite{eden2022sublinear}: for any subgraph $H$ on $k$ nodes, there is a randomized $\mathsf{CONGEST}$ algorithm that solves induced $H$ detection in $\tilde{O}(n^{2-\frac{2}{3k+1}})$ rounds.

\paragraph*{Local certification of $P_k$-freeness}

Another line of research focuses on local certification of $P_k$-freeness. In local certification, small labels called certificates are assigned to the nodes of the graph, and each node decides if the graph satisfies some property by using its local view. There are two important parameters, that is, the size of certificates and the locality of verification. In local certification of size $s$ and locality $\ell$, each node receives an $O(s)$-bit certificate from the external entity called the \textit{prover}, and each node decides its binary output (\textit{accept} or \textit{reject}) only using the unique identifiers and certificates of its $\ell$-hop neighbors. The scheme (with the locality $\ell=1$) is introduced by~\cite{korman2010proof} under the name of proof-labeling schemes, which is now a popular setting in the community of distributed computing:

\begin{definition}[proof-labeling schemes]\label{def:local-certification}
    A proof-labeling scheme ($\mathsf{PLS}_{\ell}$) of size $s$ is a pair ($c,\mathcal{A}$) of functions $c$ and $\mathcal{A}$ of the input graph $G=(V,E)$ such that
    \begin{itemize}
        \item for each node $u$, the function $c$ assigns a bit string $c(u)$ of length $s$ called the certificate of $u$;
        \item for each node $u$, the function $\mathcal{A}$, called the verification algorithm, depends on the certificates of $\ell$-hop neighbors of $u$ as well as the identifiers of these nodes, and decides the output of $u$. More precisely, let $N_{\ell}(u)=\{v_1,\ldots,v_d\}$ be the $\ell$-hop neighbors of $u$, then
        \[
        \mathcal{A}(\mathsf{id}(u),c(u),\mathsf{id}(v_1),c(v_1),\ldots,\mathsf{id}(v_d),c(v_d))
        \] is the output of $u$.
    \end{itemize}
        Let $\mathcal{P}$ be a graph property. We say that there is a $\mathsf{PLS}_{\ell}$ that certifies $\mathcal{P}$ if there exists a verification algorithm $\mathcal{A}$ with locality $\ell$ (which outputs $\mathsf{Accept}$ or $\mathsf{Reject}$) satisfies the following conditions.
    \begin{itemize}
        \item If $G$ satisfies the property $\mathcal{P}$, there exists a certificate function $c$ such that all nodes in $G$ output $\mathsf{Accept}$.
        \item If $G$ does not satisfy the property $\mathcal{P}$, for any certificate function $c$, at least one node in $G$ outputs $\mathsf{Reject}$.
    \end{itemize}
\end{definition}

When the locality of the verification algorithm is $\ell \geq \lfloor \frac{k}{2} \rfloor$, certificates are unnecessary for $P_k$-freeness (i.e., $P_k$-freeness can be solved in $\lfloor \frac{k}{2} \rfloor$ rounds of the verification algorithm, as the radius of $P_k$ is $\lfloor \frac{k}{2} \rfloor$). Thus, our interest lies in cases where the path length is long relative to the locality. Recently, Bousquet, Cook, Feuilloley, Pierron, and Zeitoun~\cite{bousquet2024local} studied this topic by analyzing the relationship among path length, certificate size, and locality of verification. They provided various nontrivial upper and lower bounds on certificate size (we discuss details in Section~\ref{subsec:our-results}). While local certification of various subgraph-related problems has been studied (e.g., ~\cite{bousquet_et_al:LIPIcs.OPODIS.2021.22,crescenzi_et_al:LIPIcs.DISC.2019.13,fraigniaud2024meta,fraigniaud2019distributed}), the only known result for local certification of $P_k$-freeness prior to~\cite{bousquet2024local} was for $k=4$: there is a $\mathsf{PLS}_1$ that certifies $P_4$-freeness with $O(\log n)$ bits~\cite{fraigniaud2023distributed}. Moreover, \cite{fraigniaud2023distributed} proved that every $\mathsf{MSO}_1$ property $\Pi$ can be certified with $O(\log^2 n)$ bits and locality 1 if all graphs satisfying $\Pi$ have bounded clique-width. 
Since $P_k$-freeness can be expressed by $\mathsf{MSO}_1$, this suggests that the difficulty gap between certifying $P_4$-freeness and $P_5$-freeness can be attributed to the fact that $P_4$-free graphs have bounded clique-width, whereas $P_5$-free graphs do not.

\subsection{Our results}\label{subsec:our-results}

\subsubsection{Result 1: $P_k$-freeness for $k=4$ and $k=5$}
We first show the following upper bounds for $k=4$ and $k=5$.
\begin{theorem}\label{thm-4-path-upper-bound}
 There exists a randomized algorithm that solves $P_4$-freeness in the $\mathsf{broadcast\;CONGEST}$ model, running in $O\left(\frac{\log n}{\log\log n}\right)$ rounds. 
\end{theorem}
\begin{theorem}\label{thm-5-path-upper-bound}
 There exists a randomized algorithm that solves $P_5$-freeness in the $\mathsf{CONGEST}$ model, running in $O\left(n\log^2 n\right)$ rounds. 
\end{theorem}

\paragraph*{Technical challenges}
As mentioned above, there were no nontrivial algorithms for $P_k$-freeness for $k\geq 4$ in the $\mathsf{CONGEST}$ model beyond the general $O(n^{2-2/(3k+2)})$ upper bound~\cite{eden2022sublinear}. Moreover, the previous $\mathsf{CONGEST}$ algorithms for subgraph detection have been mostly focused on non-induced cases. Most techniques from these results seem to be useless for induced subgraph detecion. Typical examples are the \textit{expander decomposition} applicable to cliques~\cite{chang2019distributed,chang2019improved}, and the \textit{BFS search with threshold} applicable to (non-induced) cycles~\cite{censor2020distributed,fraigniaud2024meta}.
We thus need to exploit structural properties of $P_4$/$P_5$-free graphs, resulting our algorithms quite different from the other algorithms for subgraph detection. 
We believe that our approach hints new algorithms for other induced subgraphs.\vspace{2mm}

We turn our attention to the framework of local certification. We get the following certification scheme for $P_5$-freeness. 

\begin{theorem}\label{thm:local-certification}
There is a $\mathsf{PLS}_1$ that certifies $P_5$-freeness with certificates of size $O(n\log n)$.
\end{theorem}
To obtain this result, we directly use the algorithm of Theorem~\ref{thm-5-path-upper-bound}. This is somewhat interesting, as $\mathsf{CONGEST}$ algorithms for subgraph freeness cannot be used for local certification of subgraph freeness in general (unless they are $\mathsf{broadcast\;CONGEST}$ algorithms). 
We also show that the size of our certificates is optimal, up to a subpolynomial factor: 
\begin{theorem}\label{thm:5-path-lower-bound}
For any $k\geq 4\ell + 1$, any $\mathsf{PLS}_{\ell}$ for $P_k$-freeness requires certificates of size $\Omega\left(\frac{n}{e^{O(\sqrt{\log n})}}\right)$.
\end{theorem}
This result is proved via a combination of several known reduction techniques: we first reduce the nondeterministic three-party communication complexity of set-disjointness function to triangle freeness, which is then reduced to $P_5$-freeness.

\paragraph*{Independent and concurrent works} After submitting the first draft of this paper, we learned about the following independent and concurrent works: Bousquet and Zeitoun~\cite{bousquet2024subquadratic} proved that there is a $\mathsf{PLS}_1$ that certifies $P_5$-freeness with certificates of size $O(n^{3/2})$. Additionally, as noted in~\cite{bousquet2024subquadratic}, Chaniotis, Cook, Hajebi, and Spirkl obtained $\Omega(n^{1-o(1)})$ lower bound for $P_5$-freeness.

\subsubsection{Result 2: $P_k$-freeness for general $k$}

We provide the following lower bounds for $P_k$-freeness in the $\mathsf{CONGEST}$ model using the two-party communication complexity, which is the standard framework for distributed subgraph detection problems.
\begin{theorem}\label{thm:11-path-lower-bound}
Let $d$ be any positive integer.
\begin{itemize}
    \item For $d\leq 2$, $P_k$-freeness for $k\geq 11d$ require $\tilde{\Omega}(n^{2-1/d})$ rounds in the $\mathsf{CONGEST}$ model.
    \item  For $d\geq 3$, $P_k$-freeness for $k\geq 8d$ requires $\tilde{\Omega}(n^{2-1/d})$ rounds in the $\mathsf{CONGEST}$ model.
\end{itemize}
\end{theorem}

\noindent The $n^{1/d} = n^{1/\Theta(k)}$ factor is unavoidable, in light of the upper bound of $n^{2-\Theta(1/k)}$ established by~\cite{eden2022sublinear}. We will show these lower bounds using the standard reduction from two-party communication complexity. Although the framework used here is classic, there are several challenges that makes the proof highly non-trivial, which we discuss in detail later in Section~\ref{sec:P11-freeness}.

Additionally, our constructions apply to the nondeterminictic two-party communication setting, leading to the following. 
\begin{theorem}\label{thm:local-certification-lower-bound}
Let $\ell$ and $d$ be positive integers. 
\begin{itemize}
    \item For $d \leq 2$, any $\mathsf{PLS}_\ell$ for $P_k$-freeness requires certificates of size $\tilde{\Omega}(n^{2-1/d})$ for $k\geq 4d\ell + 7d$. 
    \item  For $d \geq 3$, any $\mathsf{PLS}_\ell$ for $P_k$-freeness requires certificates of size $\tilde{\Omega}(n^{2-1/d})$ for $k\geq 4d\ell + 4d$. 
\end{itemize}
\end{theorem}


\paragraph*{Comparison with the recent results by Bousquet et al.~\cite{bousquet2024local}}
Let us now compare our results with those of Ref.~\cite{bousquet2024local}, who considered local certification of $P_k$-freeness and obtained the following: 

\begin{theorem}[\cite{bousquet2024local}]\label{thm:bousquet-algorithm}
    If the locality is $\ell \geq 1$, then:
    \begin{itemize}
        \item $P_{k}$-freeness for $k\leq 3\ell - 1$ can be certified with certificates of size $O(n\log^3n)$;
        \item $P_{k}$-freeness for $k\leq \lceil\frac{14\ell}{3}\rceil -1$ can be certified with certificates of size $O(n^{3/2}\log^2 n)$; 
        \item $P_{k}$-freeness for $k\geq 4\ell + 3$ requires certificates of size $\Omega(n/\ell)$.
    \end{itemize}
\end{theorem}

They also conjectured the following.
\begin{conjecture}[\cite{bousquet2024local}]\label{conjecture:local certification}
    For all $\alpha > 0$, the optimal size for the local certification of $P_{\alpha \ell}$-free
graphs with locality $\ell$ is of the form $\Theta(n^{2-1/f(\alpha)})$, for some unbounded increasing
function $f$.
\end{conjecture}

We summarize our results and previous results in Table~\ref{table:our results}.
For $\ell = 1$, the results of Bousquet et al.~\cite{bousquet2024local} provide nontrivial lower bound for $P_7$-freeness and no nontrivial upper bound, so they left the complexities of remaining path lengths open. In Theorems~\ref{thm:local-certification} and~\ref{thm:5-path-lower-bound}, we show the nearly optimal bound of $P_5$-freeness for $\ell = 1$. 

For general $k$, Theorem~\ref{thm:local-certification-lower-bound} shows the first superlinear lower bounds, improving the previous linear lower bounds. Moreover, Theorem~\ref{thm:local-certification-lower-bound} partially answers Conjecture~\ref{conjecture:local certification}: $f$ is at least linearly increasing.

\begin{table}[tb]\centering
  \caption{Summary of our results and previous results on local certification of $P_k$-freeness. Here $n$ denotes the number of nodes in the network.} 
  \begin{tabular}{cccc} \hline
    Path length & Certificate size & Model & Reference \\ \hline
    5 & $O(n\log n)$ &  $\mathsf{PLS}_1$ &Thm~\ref{thm:local-certification} \\ 
    $4\ell + 1$ & $\Omega(n^{1-o(1)})$  & $\mathsf{PLS}_{\ell}$ & Thm~\ref{thm:5-path-lower-bound} \\ 
    $8\ell + 14$ & $\tilde{\Omega}(n^{3/2})$ & $\mathsf{PLS}_{\ell}$ &Thm~\ref{thm:local-certification-lower-bound}\\
    $4d\ell + 4d$, $d\geq 3$ & $\tilde{\Omega}(n^{2-1/d})$ & $\mathsf{PLS}_{\ell}$ & Thm~\ref{thm:local-certification-lower-bound}\\
    $3\ell - 1$ & $O(n\log ^3 n)$  & $\mathsf{PLS}_{\ell}$ &\cite{bousquet2024local}\\
    $4\ell + 3$ & $\Omega(n)$  & $\mathsf{PLS}_{\ell}$  & \cite{bousquet2024local}\\
    $\lceil \frac{14\ell}{3}\rceil - 1$ & $O(n^{3/2}\log^2 n)$ & $\mathsf{PLS}_{\ell}$  &\cite{bousquet2024local}\\
  \end{tabular}\label{table:our results}

\end{table}

\subsubsection{Result 3: Ordered path detection and applications}
We then define and study the following problem called ordered $P_k$ detection.
\begin{definition}[ordered $P_k$ detection]
    Each node of the graph has a color from $\{1,\ldots,k\}$, and the goal is to detect an induced path that consists of edges $\{(p_i,p_{i+1})\}_{i\in \{1,\ldots,k-1\}}$ on $k$ nodes $\{p_i\}_{i\in \{1,\ldots,k\}}$ where $p_i$ is colored by $i$.
\end{definition} 
\paragraph*{Motivation}
The definition of ordered path detection may seem somewhat artificial; however, it can be motivated in a manner similar to that of multicolored path detection studied in~\cite{nikabadi2022beyond}. Algorithms for the multicolored/ordered variants of these problems with color-coding techniques~\cite{alon1995color} are often used to address the standard version. For example, the state-of-the-art $\mathsf{CONGEST}$ algorithm for detecting $2k$-node cycles~\cite{censor2020distributed,fraigniaud2024even} actually detects the ordered variant of $2k$-cycles. Consequently, lower bounds for the ordered variant reflect the limitations of algorithms that employ color-coding. Moreover, as demonstrated in Theorem~\ref{thm:circuit_complexity_barrier_3_path}, we will show that ordered path detection also provides some nontrivial insight for subgraph detection in the quantum $\mathsf{CONGEST}$ model, which is not restricted to algorithms using color-coding.    

\paragraph*{Contribution}

We first prove the following lower bound for $k=5$.
\begin{theorem}\label{thm:ordered-5-path}
 Any randomized algorithm that solves ordered $P_k$ detection for $k\geq 5$ requires $\tilde{\Omega}(n)$ rounds in the $\mathsf{CONGEST}$ model.
\end{theorem}
We then focus on ordered $P_3$ detection. Our next finding is that, similarly to triangle detection, proving non-trivial lower bounds for ordered $P_3$ detection is difficult. This result also shows a circuit complexity barrier for induced $C_4$ detection.
\begin{theorem}[Informal]\label{thm:circuit_complexity_barrier_3_path}
For any constant $\varepsilon>0$, proving any lower bound of the form $\Omega(n^{\varepsilon})$ for ordered $3$-path detection in the $\mathsf{CONGEST}$ model or $\Omega(n^{1/2+\varepsilon})$ for induced $C_4$ detection in the quantum $\mathsf{CONGEST}$ model implies super-linear lower bounds on circuit complexity for an explicit family of boolean functions.
\end{theorem}
Note that our circuit complexity barrier for induced $C_4$ detection is in the \textit{quantum} $\mathsf{CONGEST}$ model~\cite{le2018sublinear}.
We then complement this result by showing a nontrivial upper bound for induced $C_4$.

\begin{theorem}\label{thm:induced-4-cycle}
In the quantum $\mathsf{CONGEST}$ model, induced $C_4$ detection can be solved in $O(n^{3/4})$ rounds with high probability.
\end{theorem}
Previously, no nontrivial upper bound for induced $C_4$ detection was known, and our bound for induced $C_4$ detection beats a classical $\tilde{\Omega}(n)$ lower bound~\cite{legallISAAC21}.

\subsection{Related work}\label{appendix:related work}

\paragraph*{Cliques and Cycles}

Among all pattern graphs, subgraph detection in the $\mathsf{CONGEST}$ model has particularly focused on detecting cliques and cycles.
For $p$-cliques, denoted $K_p$, it is known that the round complexity of $K_p$ listing is $\tilde{\Theta}(n^{1-2/p})$ for all constant $p\geq 3$~\cite{chang2021near,censor2021tight}. The $\tilde{\Omega}(\sqrt{n})$ lower bound for $K_p$ detection for any $p\geq 4$ is also known~\cite{czumaj2020detecting}. 
For cycles, the complexity varies depending on whether the cycle is induced or not. For non-induced $k$-cycles, it is known that the round complexity of $C_k$ detection is $\tilde{\Theta}(n)$ for all odd constant $k \geq 5$~\cite{drucker2014power}. When $k$ is even, the current best upper bounds are $\tilde{O}(n^{1-2/k})$ for all $k\geq 4$~\cite{CFGLLODISC20,fraigniaud2024even}, and the best lower bounds are $\tilde{\Omega}(\sqrt{n})$ for all $k\geq 4$~\cite{korhonen_et_al:LIPIcs:2018:8625}. Induced $k$-cycles are significantly harder to detect, with lower bounds of $\tilde{\Omega}(n^{2-1/\Theta(k)})$ for all $k\geq 4$~\cite{legallISAAC21}. The best upper bound for induced $H$ detection for any pattern graph $H$ on $k$-nodes is $\tilde{O}(n^{2-1/\Theta(k)})$~\cite{eden2022sublinear}.

While some of these results are tight, such as for $K_4$ detection and $C_4$ detection (both induced and non-induced), there remains a substantial gap between upper and lower bounds for other subgraphs. For instance, the current upper bound for triangle detection is $\tilde{O}(n^{1/3})$, but no better-than-constant lower bound is known. There are barriers to proving better lower bounds, which can be categorized as follows:
\begin{itemize}
    \item \textbf{Two-Party Communication Complexity Barrier:} It is known that improvements in lower bounds cannot be achieved via reductions from two-party communication complexity, which is a common technique for proving existing lower bounds for subgraph detection in the $\mathsf{CONGEST}$ model. It is known that none of the current lower bound $\tilde{\Omega}(\sqrt{n})$ of $p$-clique detection for $p\geq 5$~\cite{czumaj2020detecting}, the current lower bound $\tilde{\Omega}(\sqrt{n})$ of $6$-cycle detection~\cite{eden2022sublinear}, and the current lower bound $\tilde{\Omega}(n)$ of induced $k$-cycle detection for $k\in\{5,6,7\}$~\cite{legallISAAC21} can be improved by any polynomial factor through such reductions.
    \item \textbf{Circuit Complexity Barrier:} Improving lower bounds for subgraph detection problems may require breakthroughs in circuit complexity. Proving lower bounds of the form $\Omega(n^\varepsilon)$ for triangle detection~\cite{eden2022sublinear}, $\Omega(n^{1/2+\varepsilon})$ for $6$-cycle detection~\cite{CFGLLODISC20}, $\Omega(n^{1-1/361^2+\varepsilon})$ for $2k$-cycle detection for $k>360$~\cite{eden2022sublinear}, or $\Omega(n^{3/5+\varepsilon})$ for $p$-clique detection for $p\geq 5$~\cite{CFLLOITCS22} would imply super-linear lower bounds on polynomial-depth circuits with constant fan-in and fan-out gates.
\end{itemize}

\paragraph*{Further related work}

A substantial body of work has investigated subgraph detection in distributed settings, with a particular focus on non-induced cases~\cite{abboud2020fooling,censor2015algebraic,censor2020sparse,censor2022deterministic,dolev2012tri,fischer2018possibilities,10.1145/3662158.3662793,fraigniaud2023power,fraigniaud2019distributed,gonen2018OPODIS,doi:10.1137/20M1326040,10.1145/3460900}. There is a useful survay of distributed subgraph detection by Censor-Hillel~\cite{censorICALP2021invited} that provides a comprehensive overview of these studies. 

After 2020, research has extended to the quantum $\mathsf{CONGEST}$ model, where each node can perform arbitrary local quantum computation and each message can contain $O(\log n)$ quantum bits instead of a bit string. Izumi, Le Gall, and Magniez~\cite{izumi2020quantum}  introduced a quantum algorithm for triangle detection that runs in $\tilde{O}(n^{1/4})$ rounds, improving upon the classical upper bound of $\tilde{O}(n^{1/3})$. Subsequently, Censor-Hillel et al.~\cite{CFLLOITCS22} improved it to $\tilde{O}(n^{1/5})$, as well as showed $\tilde{O}(n^{1-2/(p-1)})$ upper bounds for $p$-clique ($p\geq 7$). For (non-induced) $k$-cycles, van Apeldoorn and de Vos~\cite{AVPODC22} showed $\tilde{O}(n^{\frac{1}{2}- \frac{1}{2k+2}})$ upper bounds for $k\in\{4,6,8,10\}$ which led to the first improvement from the current $\tilde{O}(n^{1-2/k})$ classical upper bounds. Recently, \cite{fraigniaud2024even} improved them to $\tilde{O}(n^{\frac{1}{2}- \frac{1}{k}})$ for any even $k\geq 4$.

In addition to subgraph detection, quantum algorithms have been developed for other problems in the quantum $\mathsf{CONGEST}$ model~\cite{AVPODC22,elkin2014can,fraigniaud2024even,le2018sublinear,wu2022quantum}, and there are results for other quantum distributed models and quantum local certification~\cite{akbari2024online,coiteux2024no,fraigniaud2021distributed,legallSTACS2019,izumi2019quantum,a16070332,le2022nontrivial,le2023MFCS,le2023STACS}.

\section{Preliminaries}\label{sec:preliminaries}
We write $[n]=\{1,2,\ldots,n\}$.
For a graph $G=(V(G),E(G))$, we say that $H=(V(H),E(H))$ is a subgraph of $G$ if there is an injective function $\phi:V(H)\rightarrow V(G)$ such that $(u,v)\in E(H) \Rightarrow (\phi(u),\phi(v))\in E(G)$ for any pair of nodes $u,v\in V(H)$.
We say that $H$ is an induced subgraph of $G$ if there is an injective function $\phi:V(H)\rightarrow V(G)$ such that $(u,v)\in E(H) \Leftrightarrow (\phi(u),\phi(v))\in E(G)$ for any pair of nodes $u,v\in V(H)$. Let $P_k$ be a path on $k$ nodes. We say that $G$ is $P_k$-free if $G$ does not contain $P_k$ as an induced subgraph. Throughout the paper, we assume that the input graph is connected, otherwise we can treat each connected component separately.

\paragraph*{The $\mathsf{CONGEST}$ model and variants}
In the $\mathsf{CONGEST}$ model~\cite{peleg2000distributed}, each node of the network $G=(V,E)$ has a distinct $O(\log n)$-bit identifier and can communicate with its neighbors in a synchronized manner. In each round each node (1) does some local computation and (2) sends an $O(\log n)$-bit message to each of its neighbors. In the initial state, each node only knows its own input and the identifiers of its neighbors. The $\mathsf{broadcast\;CONGEST}$ model is a weaker model in the sense that each node can only broadcast a $O(\log n)$-bit message in each round.
In the $\mathsf{congested\;clique}$ model~\cite{lotker2003mst}, we allow all-to-all communication in each round. That is, the communication topology is always the $n$-node clique, and the input graph $G$ is a subgraph of the clique. In the quantum $\mathsf{CONGEST}$ model~\cite{le2018sublinear}, each node can send $O(\log n)$-qubit quantum message to each of its neighbors, instead of $O(\log n)$-bit classical message.




\section{Algorithm for $P_4$-freeness}\label{sec:P4-freeness}

In this section we prove Theorem~\ref{thm-4-path-upper-bound}. 

It is well known that a graph is $P_4$-free iff it is a cograph~\cite{corneil1981complement}.
A cograph is defined as a graph that is constructed using the following rules:
\begin{itemize}
    \item A single-node graph $K_1$ is a cograph;
    \item For cographs $G=(V_G,E_G)$ and $H=(V_H,E_H)$, its disjoint union $(V_G \uplus V_H, E_G \uplus E_H)$ is a cograph;
    \item For cographs $G=(V_G,E_G)$ and $H=(V_H,E_H)$, its join $(V_G \uplus V_H, E_G \uplus E_H \uplus (V_G\times V_H))$ is a cograph;
\end{itemize}
Using this characterization, we can construct a low-congestion spanning tree for any cograph as follows. A similar property on the spanning tree of cographs is also used in distributed (interactive) proofs for $P_4$-freeness~\cite{fraigniaud2023distributed,montealegre2021compact}. We use the following results on the balls into bins problem.

\begin{lemma}[Balls into bins problem~\cite{gonnet1981expected}]\label{balls-into-bins}
    Assume that there are $n$ bins and $n$ balls. Each ball uniform randomly selects one bin which it places into.
    Then, with probability at least $1-1/n$, all bins have at most $O\left(\frac{\log n}{\log \log n}\right)$ balls.
\end{lemma}

\begin{lemma}
There exists a $\mathsf{broadcast\;CONGEST}$ algorithm that runs in $O(1)$ rounds and performs the following:
\begin{itemize}
\item If some node rejects, then $G$ is not a connected cograph.
\item Otherwise, with high probability, it outputs a rooted spanning tree of depth 2, where each node, except the root, has at most $O\left(\frac{\log n}{\log \log n}\right)$ children.
\end{itemize}
\end{lemma}
\begin{proof}
We first reject if the diameter of the graph is not $O(1)$, since every connected cograph has diameter at most 2. This can be accomplished in $O(1)$ rounds by running a BFS search from the node with minimum identifier as the depth of the BFS tree provides a $2$-approximation of the diameter.
From the definition of cographs, a connected cograph $G=(V,E)$ is constructed from the join operation on two distinct node sets $V_1$ and $V_2$ where $V=V_1\uplus V_2$ and $|V_1|\geq |V_2|$. 
Let $r$ be the node with the maximum degree in $G$. Note that $r$ can be found in $O(1)$ rounds since the diameter of the graph is constant. If $\mathrm{deg}(r) < n/2$, the algorithm rejects (if $G$ is a connected cograph, then the maximum degree is at least $n/2$ as each node in $V_2$ has degree at least $\abs{V_1}\geq n/2$). Consider the tree rooted at $r$ where all neighbors of $r$ are at depth 1. As below, we can construct the desired spanning tree with high probability if the graph is a connected cograph. 
\begin{itemize}
    \item {Case 1: $r\in V_2$.} Define $V_2' = V_2 \backslash (N(r)\cup \{r\})$. Observe that $V_1 \times V'_2 \subseteq E$ and $|V'_2| \leq |V_1|$. Each $v \in V_2'$ uniform randomly selects one node from its neighbors as its parent in the tree. Now, each depth-1 node $u$ (bin) is selected as the parent of $v$ (ball) with probability at most $1/\abs{V'_2}$ as $\mathrm{deg}(v)\geq \abs{V_1}$. Then from Lemma~\ref{balls-into-bins}, all depth-1 nodes have at most $O\left(\frac{\log n}{\log \log n}\right)$ children with high probability. 
    \item {Case 2: $r\in V_1$.} Define $V_1' = V_1 \backslash (N(r)\cup \{r\})$. Since $\mathrm{deg}(r)\geq |V_1|$, we have
    \begin{align*}
        |V'_1| = |V_1| - (\mathrm{deg}(r) - |V_2|) - 1 \leq |V_2|.
    \end{align*}
    Each node $v$ in $V'_1$ selects one of its neighbors as its parent uniform randomly.  Now, each depth-1 node $u$ (bin) is selected as the parent of $v$ (ball) with probability at most $1/\abs{V_2}$ as $\mathrm{deg}(v)\geq \abs{V_2}$. Then from Lemma~\ref{balls-into-bins}, all depth-1 nodes have at most $O\left(\frac{\log n}{\log \log n}\right)$ children with high probability. 
\end{itemize}

\end{proof}

\begin{proof}[Proof of Theorem~\ref{thm-4-path-upper-bound}]
 Ref.~\cite{kari2015solving} demonstrated a protocol in the randomized multiparty simultaneous messages model (also known as the distributed sketching model), where each node sends a message of size $O(\log n)$ to a referee. The referee can then determine if $G$ is a cograph with high probability. We can simulate this protocol in the constructed depth-2 spanning tree, with the root $r$ acting as the referee. Since each depth-1 node has at most $O\left(\frac{\log n}{\log \log n}\right)$ children, sending all messages to $r$ is completed in $O\left(\frac{\log n}{\log \log n}\right)$ rounds.  
\end{proof}

\section{Algorithm and certification for $P_5$-freeness} \label{sec:P5-freeness-upper-bound}
In this section we prove Theorem~\ref{thm-5-path-upper-bound} and Theorem~\ref{thm:local-certification}.

\paragraph*{A High-Level Overview of Our $\mathsf{CONGEST}$ Algorithm}
Our algorithm aims to efficiently detect whether a given graph is $P_5$-free, focusing on cases where the graph has a diameter of 2 or 3 (observe that if the diameter is at least 4, then it must not be $P_5$-free). We treat these cases separately in Section~\ref{section:diameter-2-case} and Section~\ref{section:diameter-3-case} due to their distinct properties.\\

\noindent\textbf{Case 1: Diameter 2.} 
If the input graph has a diameter of 2, our approach leverages the fact that a certain node can gather complete information about the graph in $\tilde{O}(n)$ rounds. Specifically, the node with the maximum degree, denoted as $r$, can collect all the edges in the graph. This is feasible because:
\begin{enumerate}
    \item The high degree of $r$ ensures sufficient communication bandwidth to receive $\Omega(m)$ edges in $O(n)$ rounds.
    \item By using the balanced edge partition technique~\cite{chang2019distributed,censor2021tight}, we can distribute $O(n)$ different edges to each neighbor of $r$.
\end{enumerate}
In the first $O(n)$ rounds, each neighbor of $r$ receives $O(n)$ different edges thanks to the diameter constraint. In the subsequent $O(n)$ rounds, all the edges in the graph are sent to $r$, allowing it to perform a complete $P_5$ detection. A noteworthy point here is that this procedure can be used for $H$-freeness for any pattern graph $H$, and the condition that the diameter is at most 2 seems to be crucial -- for any constant $\varepsilon > 1$, there is a pattern graph $H$ of constant size such that $H$-freeness requires $\Omega(n^{2-\varepsilon})$ rounds even when the diameter of the input graph is 3~\cite{fischer2018possibilities}.\\

\noindent\textbf{Case 2: Diameter 3.}
In the case of a graph with diameter 3, the larger diameter necessitates more tricky definitions and case analysis, resulting in a more involved proof. We outline the main ideas of our proof here before giving details.

Let $V_i$ be the set of nodes at distance $i$ from $r$, for $i \in \{1, 2, 3\}$. The edges (except the ones incident to $r$) are partitioned as $E_{i,j} = E \cap (V_i \times V_j)$ for $i \leq j$. We employ a procedure $\mathcal{P}_{\mathsf{collect}}$ (Algorithm~\ref{algorithm_collect}) that runs in $O(n)$ rounds, enabling $r$ to gather $E_{1,1}$ and $E_{1,2}$. After this procedure, if $r$ detects any missing edge in $(E_{2,2}\backslash F_{bad}) \cup E_{2,3}$—where $F_{bad}$ is a carefully defined subset of $E_{2,2}$—it can safely conclude that the graph is not $P_5$-free (Lemma~\ref{lem:reject-condition}). Therefore, if $r$ does not reject, it learns $E\backslash (F_{bad} \cup E_{3,3})$. 

For any $P_5$ containing an edge from $F_{bad} \cup E_{3,3}$, detection is performed by nodes connected to the endpoints of these edges (Lemma~\ref{lem:Fbad-path-detection} and Lemma~\ref{lem:E33-path-detection}). For any $P_5$ composed solely of edges from $E\backslash (F_{bad} \cup E_{3,3})$, detection is performed by node $r$, which now has access to $E\backslash (F_{bad} \cup E_{3,3})$. The key challenge here is distinguishing between proper $P_5$s (induced by the entire edge set $E$) and improper $P_5$s (induced by $E\backslash (F_{bad} \cup E_{3,3})$ but not by $E$). We address this by having $r$ count the improper patterns and subtract this count from the total number of $P_5$s it finds (Lemma~\ref{lem:5-node-graph-counting-with-bad-edges}). This yields the correct count of proper $P_5$s, ensuring accurate detection.

\subsection{Subgraph freeness in graphs with diameter two}\label{section:diameter-2-case}
Here we assume that the graph $G=(V,E)$ with $n$ nodes and $m$ edges has diameter 2.
Let $r\in V$ be the node with maximum degree. Then $\mathrm{deg}(r) = \Delta = \Omega(m/n)$.
Divide the node set $V=\{r\}\cup V_1 \cup V_2$ where $V_i$ is the set of nodes whose distance to $r$ is $i$. We first consider that every node $u$ broadcasts the list of its neighbors in $O(n)$ rounds. After that $r$ learns all the edges connected to $\{r\}\cup V_1$.
In the remaining part we show that, with high probability, $r$ can learn the edge set $E\cap (V_2\times V_2)$ in another $O(n)$ rounds. Therefore, $r$ can locally decide if the graph is $P_5$ free. 

Assume that we divide $V_2$ into $\sqrt{m/n}$ subsets $V_2^i$ for $i\in [\sqrt{m/n}]$ as follows (assuming $\sqrt{m/n}$ is an integer): each node $v\in V_2$ chooses an integer $i\in [\sqrt{m/n}]$ uniformly at random and joins $V_2^i$. We then label the set $V_1$ as $\{v_j|j\in [\Delta] \}$. All the integers $i$ and the label of $V_1$ nodes are informed to all the nodes in the network in $O(n)$ rounds.
We use the following lemma.
\begin{lemma}[\cite{chang2019distributed}]\label{lemma_partition}
    Consider a graph with $\bar{m}$ edges and $\bar{n}$ vertices. We generate a subset $S$ by letting each
vertex join $S$ independently with probability $p$.
Suppose that the maximum degree is $\Delta\leq \bar{m}p/20\mathrm{log}\bar{n}$ and $p^2\bar{m}\geq 400\mathrm{log}^2\bar{n}$. 
Then, with probability $1-1/n^c$ for sufficiently large $c$, the number of edges in the subgraph induced by $S$ is at most
$6p^2 \bar{m}$.
\end{lemma}
Here we set 
$\bar{n}=n$, $\bar{m}=m$, $p=\sqrt{n/m}$. Clearly, $p^2\bar{m}\geq 400\mathrm{log}^2\bar{n}$ holds. As $\Delta < n$, one can verify that $\Delta\leq \bar{m}p/20\mathrm{log}\bar{n}$ holds if $\bar{m}=m\geq 400 n\log^2 n$. We assume that $m\geq 400 n\log^2 n$ holds since otherwise a simple $O(m)=\tilde{O}(n)$-round algorithm suffices. Then, with high probability, $|E\cap V_2^i\times V_2^j| = O(p^2 m)=O(n)$ for all $i,j \in [\sqrt{m/n}]$. Consider the node $v_k \in V_1$ collects the edge set $E\cap V_2^i\times V_2^j$ for $k = i + (j-1)\sqrt{m/n}$. To do this, assume that $(u,v)\in E\cap V_2^i\times V_2^j$. Then there exists a node $w$ such that $(u,w),(w,v_k)\in E$ since the diameter of the graph is 2. $w$ knows the edge $(u,v)$ because of the previous $O(n)$ round communication and thus $w$ can send it to $v_j$. Since $|E\cap V_2^i\times V_2^j| = O(n)$, this is done in $O(n)$ rounds. Finally, $v_k$ sends $E\cap V_2^i\times V_2^j$ to $r$ in $O(n)$ rounds. Now we have established the following theorem.

\begin{theorem}\label{thm-diameter-2-upper-bound}
 In graphs with diameter 2, for any subgraph $H$, (non-induced and induced) $H$-freeness can be solved in $O(n\log^2 n)$ rounds in the $\mathsf{CONGEST}$ model with high probability.
\end{theorem}
It is noteworthy that this result does not hold when the diameter is greater than 2, since for any constant $\varepsilon>0$, there is a pattern graph $H$ such that $H$-freeness in graphs with diameter 3 requires $\Omega(n^{2-\varepsilon})$ rounds~\cite{fischer2018possibilities,legallISAAC21}.

\subsection{$P_5$-freeness in graphs with diameter three}\label{section:diameter-3-case}
Assume that the diameter of the input graph (with $n$ nodes and $m$ edges) is 3 and $m\geq 400 n\log^2 n$. Similarly to the previous case, we divide the node set into $V=\{r\}\cup V_1 \cup V_2 \cup V_3$ where $V_i$ is the set of nodes whose distance to $r$ is $i$. For all $i,j\in \{1,2,3\}$, $i\leq j$, we also define $E_{i,j}$ by the edge set between $V_i$ and $V_j$. We use a subprotocol $\mathcal{P}_{\mathsf{collect}}$ described in Algorithm~\ref{algorithm_collect}, and analyze it here.

\begin{algorithm}[ht!]
$r$: the node with maximum degree.\\
$T$: the rooted spanning tree created by the BFS traversal from $r$.
\begin{description}
    \item[Step 1] Each node broadcasts the list of its neighbors in $O(n)$ rounds.
    \item[Step 2] Consider the partition of $V$ into $\sqrt{m/n}$ subsets $V^i$ for $i\in [\sqrt{m/n}]$ as follows: Each node $v\in V$ chooses an integer $i\in [\sqrt{m/n}]$ uniformly at random, joins $V^i$ and tells it to all other nodes. The node $r$ assigns an integer $j\in [\Delta]$ to each of its neighbor, that is, create labels $V_1 = \{v_j:j\in[\Delta]\}$ where $v_j$ is a neighbor of $r$ such that the assigned integer is $j$. Then these labels as well as the partition $\mathcal{V}=\{V^i\}_{i\in [\sqrt{m/n}]}$ are shared by all the nodes in $O(n)$ rounds using pipelining.
    \item  [Step 3]
    Define $\tilde{E}_{2,2}$ by the set of edges $(u,v)\in E_{2,2}$ satisfying at least one of the following: 
    \begin{enumerate}
        \item $N(v_k) \cap (N(u)\cup N(v))\neq \emptyset$ where $u\in V^i$ and $v\in V^j$ for the partition $\mathcal{V}$, and $v_k$ is a neighbor of $r$ such that the integer $k$ is assigned in Step 2 for $k = i + (j-1)\sqrt{m/n}$,
        \item $(N(u)\cup N(v))\cap V_3 \neq \emptyset$.
    \end{enumerate}
  
    We also define $F_{bad}$ by the set of edges $(u,v) \in E_{2,2} \backslash \tilde{E}_{2,2}$ satisfying $N(u)\cap V_1 = N(v)\cap V_1$. Each node $u$ tells the neighbors the list of its incident edges in $\tilde{E}_{2,2}$ and $F_{bad}$. 

    \item[Step 4]
    For each edge $(v_k,w)$, $w$ sends the edge $(u,v)$ between $V^i$ and $ V^j$ to $v_k$ if $w$ received $(u,v)$ in Step 1. After that, each $v_k\in V_1$ sends these edges between $V^i$ and $ V^j$ to $r$. 

    \item[Step 5] $r$ computes  $\abs{E_{2,3}},\abs{\tilde{E}_{2,2}},\abs{F_{bad}}$. This is done by using the tree $T$. $r$ then rejects if the following holds:
    \begin{description}
        \item[Condition 1] There is an edge in $E_{2,3}\cup \tilde{E}_{2,2}$ that is not sent to $r$.
    \end{description}
    For each edge $(u,v)\in E_{2,2} \backslash \tilde{E}_{2,2}$, $u$ rejects if the following holds: 
    \begin{description}
        \item[Condition 2] $N(u)\cap V_1 \neq N(v)\cap V_1$.
    \end{description}
    For each edge $(u,v)\in F_{bad}$, $u$ rejects if the following holds: 
    \begin{description}
        \item[Condition 3] There is a node $w \in N(u)\cap V_2$ such that $N(u)\cap V_1 \nsubseteq N(w)$.
    \end{description}
    For each edge $(u,v)\in E_{3,3}$, $u$ rejects if the following holds: 
    \begin{description}
        \item[Condition 4] $N(u)\cap V_2 \neq N(v)\cap V_2$.
    \end{description}
\end{description}
		\caption{$\mathcal{P}_{\mathsf{collect}}$}\label{algorithm_collect}
\end{algorithm}

\begin{lemma}\label{lem:reject-condition}
    If any of the conditions in $\mathcal{P}_{\mathsf{collect}}$ (Conditions 1-4) are met, then $G$ contains an induced $P_5$.
\end{lemma}

\begin{proof}
    \begin{description}
    \item[Condition 1] Assume that there is an edge in $E_{2,3}$ that is not sent to $r$ in $\mathcal{P}_{\mathsf{collect}}$. Then for some $i,j,k$ satisfying $k = i + (j-1)\sqrt{m/n}$, there exists an edge $(u,v)$ between $V^i$ and $ V^j$ such that no node in $N(u)\cup N(v)$ is connected to $v_k$. W.l.o.g., we assume that $u\in V_2$ and $v\in V_3$. Then there exists a node $w\in N(u)\cap V_1$ that is not connected to $v_k\in V_1$. $\{v_k,r,w,u,v\}$ induces a $P_5$.

    Assume that there is an edge in $\tilde{E}_{2,2}$ that is not sent to $r$ in $\mathcal{P}_{\mathsf{collect}}$. First, observe that if an edge $(u,v)$ in $\tilde{E}_{2,2}$ satisfies the first condition of $\tilde{E}_{2,2}$, it is sent to $r$ in $\mathcal{P}_{\mathsf{collect}}$. Assume that $(u,v)$ satisfies the second condition of $\tilde{E}_{2,2}$ and is not sent to $r$ in $\mathcal{P}_{\mathsf{collect}}$. Then w.l.o.g., there are two nodes $w\in N(u)\cap V_3$ and $x\in N(u)\cap V_1$. Furthermore, assume that $u\in V^i$, $v\in V^j$, and $k = i + (j-1)\sqrt{m/n}$. $\{v_k,r,x,u,w\}$ induces a $P_5$.

    \item[Condition 2] Suppose that, for some $i,j,k$ satisfying $k = i + (j-1)\sqrt{m/n}$, there exists an edge $(u,v)\in E_{2,2} \backslash \tilde{E}_{2,2}$ between $V^i$ and $ V^j$ satisfying $N(u)\cap V_1 \backslash N(v)\cap V_1 \neq \emptyset$. Then there exist a node $w\in V_1 \cap (N(u)\backslash N(v))$ that is not connected to $v_k\in V_1$. We can conclude that $\{v_k,r,w,u,v\}$ is an induced $P_5$.

    \item[Condition 3] Let $(u,v)\in F_{bad} \cap (V^i \times V^j)$. For a node $w \in (N(u)\cup N(v))\cap V_2$ such that $N(u)\cap V_1 \nsubseteq N(w)$, if there is $y\in (N(u)\cap V_1)\backslash N(w)$, then $\{v_k,r,y,u,w\}$ is an induced $P_5$ for $k = i + (j-1)\sqrt{m/n}$.

    \item[Condition 4] Let $w\in (N(u)\backslash N(v)) \cap V_2$. Then there exists a node $x\in N(w)\cup V_1$. $\{r,x,w,u,v\}$ is a $P_5$.
    \end{description}
    
\end{proof}

Now, assuming that $\mathcal{P}_{\mathsf{collect}}$ does not reject, the following properties hold for each edge $(u,v)\in F_{bad}$ between $V^i$ and $V^j$.
\begin{description}
    \item[Property 1] $N(u)\cap V_1 = N(v)\cap V_1$,
    \item[Property 2] $N(v_k) \cap (N(u)\cup N(v)) = \emptyset$ for $k = i + (j-1)\sqrt{m/n}$,
    \item[Property 3] $(N(u)\cup N(v))\cap V_3 = \emptyset$,
    \item[Property 4] for each $w \in (N(u)\cup N(v))\cap V_2$, $N(u)\cap V_1 = N(v)\cap V_1 \subseteq N(w)$.
\end{description}

\begin{lemma}\label{lem:Fbad-path-detection}
    Assume that no node rejects in $\mathcal{P}_{\mathsf{collect}}$.
    If $G$ contains an induced $P_5$ with at least one edge from $F_{bad}$, then there is a node that detects an induced $P_5$.
\end{lemma}

\begin{proof}
    We write $\{(p_i,p_{i+1})\}_{i\in \{1,2,3,4\}}$ be the four edges constructing some induced $P_5$.
    Here we assume that exactly one edge of them is bad. First, consider that $(p_2,p_3)\in F_{bad}$. Fix one node $v\in N(p_1)\cap N(p_2)\cap V_1$. Then from Property 1 and Property 3, $p_4$ is in $V_2$. Moreover, $p_4\in N(v)$ from Property 4. Similarly, $p_1\in V_2 \cap N(v)$. Now observe that $v$ knows $\{(p_i,p_{i+1})\}_{i\in \{1,2,3,4\}}$ and $v$ can detect the path.

    We then consider that $(p_1,p_2)\in F_{bad}$. Similarly to the above case, we can assume that $p_3\in V_2$. We distinguish the following cases.
    \begin{description}
    \item[Case 1: $p_4\in V_2$.] If $p_5\in V_1$, $p_3$ can detect the path since $p_3$ knows $p_5\notin N(p_2)$ and thus $p_5\notin N(p_1)$ due to Property 1. 
    If $p_5\in V_3$, $p_3$ can detect the path due to Property 3.
    If $p_5\in V_2$, for each node $v\in N(p_1)\cap V_1$, $v\in N(p_2)\cap N(p_3)$. If $(p_4,v)\in E$ or $(p_5,v)\in E$, $v$ can detect the path. Otherwise $\{p_1,v,p_3,p_4,p_5\}$ is an induced $P_5$. $p_3$ can detect this path as follows. Since $p_3$ knows $(p_1,p_2)\in F_{bad}$ and $(v,P_5)\notin E$, $p_3$ can conclude that $(p_1,p_5)\notin E$ due to Property 4.
    
    \item[Case 2: $p_4\in V_1$.] If $p_5\in V_1$, then $p_3$ can detect the path due to Property 1. Assume $p_5\in V_2$. Fix arbitrary node $v\in N(p_1)\cap N(p_2)\cap N(p_3)\cap V_1$. If $p_5\in N(v)$, then $v$ can detect the path. Otherwise, $p_3$ knows that $p_5\notin N(v)$ and $(p_2,p_5)\notin E$. $p_3$ can detect the path since $p_3$ can verify that $(p_1,p_5)\notin E$ due to Property 4.
    
    \item[Case 3: $p_4\in V_3$.]  If $p_5\in V_3$, then $p_3$ can detect the path since $p_1,p_2$ do not have a neighbor in $V_3$. Assume $p_5\in V_2$. Fix arbitrary node $v\in N(p_1)\cap N(p_2)\cap N(p_3)$. If $p_5\in N(v)$, then $v$ can detect the path. Otherwise, $p_3$ knows that $p_5\notin N(v)$. From Property 4, we have $p_5\notin N(p_1)\cup N(p_2)$. $p_3$ can detect the path.
    \end{description}
    Therefore, any induced $P_5$ that contains exactly one edge from $F_{bad}$ can be detected by some node.

    Now we consider a path $\{(p_i,p_{i+1})\}_{i\in \{1,2,3,4\}}$ containing at least two edges from $F_{bad}$. If there are two consecutive bad edges in the path, $(p_i,p_{i+1}),(p_{i+1},p_{i+2})$, there is a node $u\in V_1\cap N(p_{i})\cap N(p_{i+1})\cap N(p_{i+2})$ from Property 1. $u$ is also connected to $p_{i+3}$ when $i\leq 2$ or $p_{i-1}$ when $i\geq 2$. $u$ can detect the path. 
The remaining case is the path containing two bad edges that are not consecutive.
We have two distinct cases.

\begin{description}
    \item[Case 4: $(p_1,p_2),(p_3,p_4)\in F_{bad}$.] From Property 1 and Property 4, there is a node $u\in N(p_1)\cap N(p_2)\cap N(p_3) \cap N(p_4)$ who can detect the path.
    \item[Case 5: $(p_1,p_2),(p_4,p_5)\in F_{bad}$.] This is similar to a subcase of Case 1. For each node $v\in N(p_1)\cap V_1$, $v\in N(p_2)\cap N(p_3)$. If $(p_4,v)\in E$ or $(p_5,v)\in E$, $v$ can detect the path. Otherwise $\{p_1,v,p_3,p_4,p_5\}$ is an induced $P_5$. $p_3$ can detect this path as follows. Since $p_3$ knows $(p_1,p_2)\in F_{bad}$ and $(v,P_5)\notin E$, $p_3$ can conclude that $(p_1,p_5)\notin E$ due to Property 4.
\end{description}

\end{proof}

\begin{lemma}\label{lem:E33-path-detection}
    Assume that no node rejects in $\mathcal{P}_{\mathsf{collect}}$.
    If $G$ contains an induced $P_5$ with at least one edge from $E_{3,3}$, then there is a node that detects an induced $P_5$.
\end{lemma}
\begin{proof}
    We write $\{(p_i,p_{i+1})\}_{i\in \{1,2,3,4\}}$ be the four edges constructing some induced $P_5$.
    Consider that $(p_2,p_3)\in E_{3,3}$ for instance. $p_1,p_4\in V_3$ since Condition 4 of $\mathcal{P}_{\mathsf{collect}}$ does not hold. Similarly, since $(p_3,p_4)\in E_{3,3}$ we have $p_5\in V_3$. Therefore there exists $v\in N(p_1)\cap N(p_2)\cap N(p_3)\cap N(p_4)\cap N(p_5)\cap V_2$ that can detect the path. The analysis of the other cases is done in the same way -- we can show that all path edges are from $E_{3,3}$.
\end{proof}
Given Lemmas~\ref{lem:Fbad-path-detection} and~\ref{lem:E33-path-detection}, we can detect any induced $P_5$ that involves edges from $F_{bad}\cup E_{3,3}$. 

Our focus now shifts to the detection of an induced $P_5$ solely composed of edges from $E\backslash (F_{bad}\cup E_{3,3})$. In this context, we need to ensure that any $P_5$ induced by $E\backslash (F_{bad}\cup E_{3,3})$ does not actually arise as an artifact of the removal of $F_{bad}\cup E_{3,3}$, but is genuinely induced by $E$. To address this, we consider all five-node induced subgraphs of $G$. As illustrated in Figure~\ref{fig:5-node-graphs-all}, there are exactly 16 distinct five-node patterns $\mathcal{H}=\{H_1,\ldots , H_{16}\}$, each containing a $P_5$ as a subgraph. See Figure~\ref{fig:5-node-graphs} for the illustration. 

We introduce the concept of a ``dangerous'' induced subgraph. An induced copy of $H\in \mathcal{H}$ in the input network $G=(V,E)$ is \textit{dangerous} iff it induces a $P_5$ in $E\backslash (F_{bad}\cup E_{3,3})$.
\begin{lemma}\label{lem:5-node-graph-counting-with-bad-edges}
Let $\mathcal{H} = \{H_i\}_{i\in [16]}$ be a set of five-node graphs that contain $P_5$ as a (non-necessarily induced) subgraph as in Figure~\ref{fig:5-node-graphs-all}. Then, for any $H\in \mathcal{H}$, the number of induced copies of $H$ in the graph that are dangerous can be counted (by the node $r$) in $O(n)$ rounds.
\end{lemma}
\begin{proof}
    We first note that there is no induced copy of $H\in \mathcal{H}$ including edges from both $F_{bad}$ and $E_{3,3}$, since two endpoints of a bad edge do not have neighbors in $V_3$ due to Property 3 of $F_{bad}$.
    
    We consider the case of $H_1$ in Figure~\ref{fig:5-node-graphs-all}, i.e., 5-node cycles $C_5$. Let $$\{(c_1,c_2),(c_2,c_3),(c_3,c_4),(c_4,c_5),(c_5,c_1)\}$$ be five edges that form a dangerous $C_5$ where $(c_5,c_1)\in F_{bad}$ and $(c_i,c_{i+1})\in E\setminus (F_{bad}\cup E_{3,3})$ for $i\in\{1,2,3,4\}$. Then, from Property 1 and Property 3 of $F_{bad}$, we can assume that $c_2,c_4\in V_2$. Moreover, from Property 4 all the nodes in $N(c_1)\cap V_1 = N(c_5)\cap V_1$ are neighbors of $c_2$ and $c_4$. Let $u$ be the node in $N(c_1)\cap V_1 = N(c_5)\cap V_1$ with minimum identifier. Then $u$ can detect the cycle as it knows all the neighbors of $c_1,c_2,c_4,c_5$. 
    Each node counts the number of such cycles, and sends it to $r$. In this way, $r$ can count the number of dangerous $C_5$ in the graph. Note that it is impossible to have $(c_5,c_1)\in E_{3,3}$ due to the fact that Condition 4 of $\mathcal{P}_{\mathsf{collect}}$ does not hold: If $(c_5,c_1)\in E_{3,3}$, then $c_2,c_4\in V_3$. It contradicts the assumption that $(c_1,c_2)\notin E_{3,3}$.

    For any other pattern graph from $\{H_{2},\ldots,H_{16}\}$ illustrated in Figure~\ref{fig:5-node-graphs-all}, we can see that at least one node of the pattern is connected to at least three other nodes of the pattern who can detect the pattern. This completes the proof.
\end{proof}

\begin{figure}[htbp]
    \begin{center}
        \includegraphics[scale=0.6]{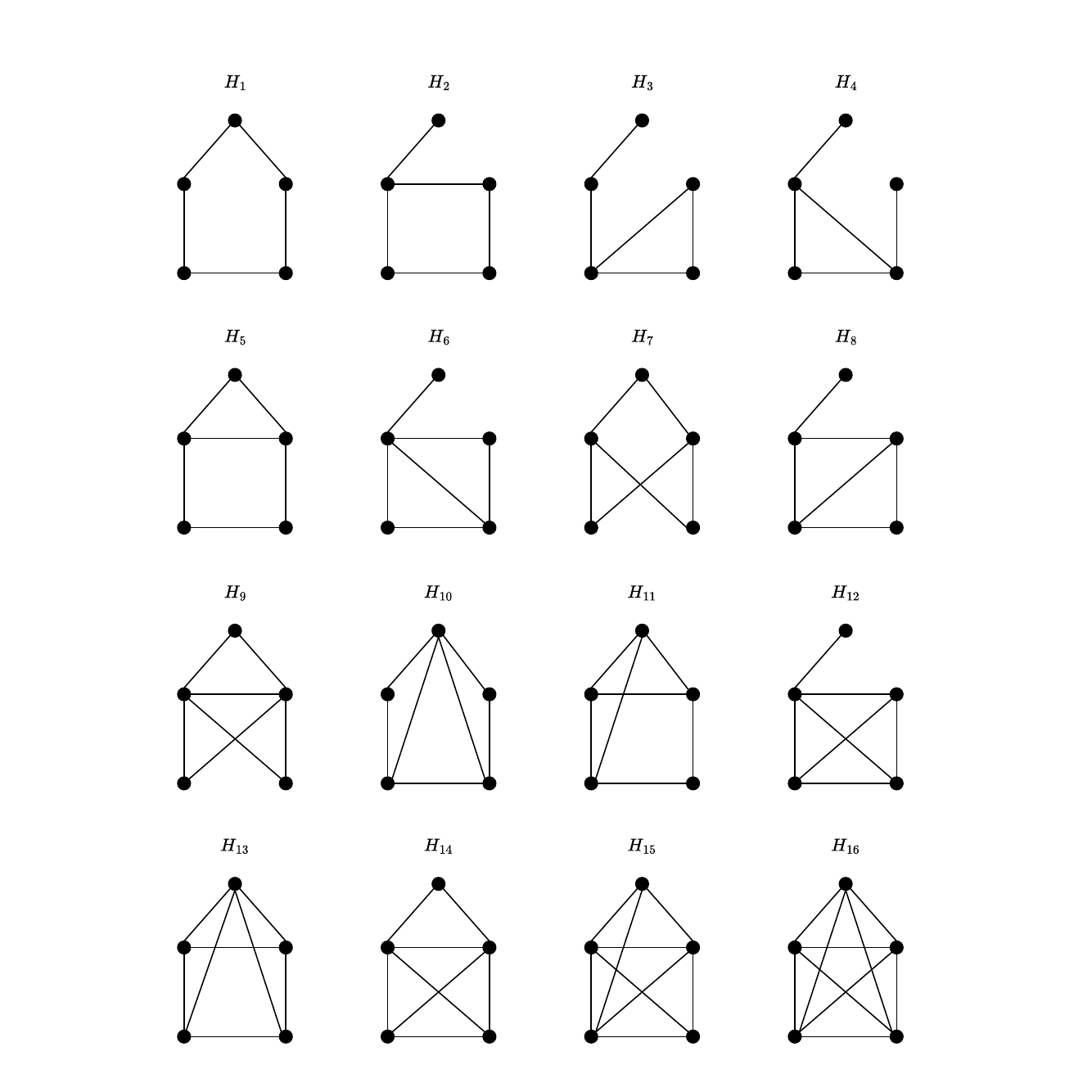}
    \end{center} 
\caption{All different 5-node graphs that contain a $P_5$ as a subgraph.}
\label{fig:5-node-graphs-all}
\end{figure}

\begin{figure}[tbp]
    \begin{center}
        \includegraphics[scale=0.6]{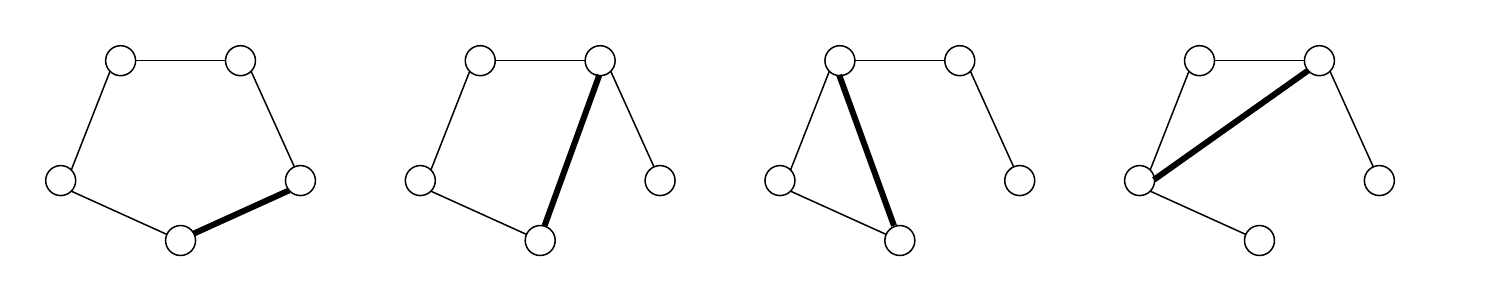}
    \end{center} 
\caption{An illustration of $H$ that contains exactly one bad edge considered in Lemma~\ref{lem:5-node-graph-counting-with-bad-edges}. Thick edges represent bad edges. For instance, an induced $C_5$ (the leftmost graph) contains exactly one edge from $F_{bad}\cup E_{3,3}$, and the remaining edges from $E\backslash (F_{bad}\cup E_{3,3})$, then removing $F_{bad}\cup E_{3,3}$ from the graph creates a new induced $P_5$.}
\label{fig:5-node-graphs}
\end{figure}

\begin{proof}[Proof of Theorem~\ref{thm-5-path-upper-bound}]
In $O(n)$ rounds, we can compute the diameter of the graph~\cite{holzer2012optimal}. If the diameter exceeds 3, the network immediately rejects. If the diameter is two, we apply Theorem~\ref{thm-diameter-2-upper-bound}, which establishes that $P_5$-freeness can be verified in $O(n\log^2 n)$ rounds with high probability. So we focus on the case where the diameter is exactly 3.

First, if the number of edges in the graph satisfies $m\leq 400n\log^2 n$, the node $r$ collects all edges in $O(m)$ rounds and decide if $G$ is $P_5$-free.
Otherwise, the network runs $\mathcal{P}_{\mathsf{collect}}$ in $O(n)$ rounds. Assuming $\mathcal{P}_{\mathsf{collect}}$ does not reject, $r$ knows all the edges in $E\backslash (F_{bad}\cup E_{3,3})$ as $E_{2,2}\backslash \tilde{E}_{2,2} = F_{bad}$.

Next, we assume that there is no induced $P_5$ that contains at least one edge from $F_{bad}\cup E_{3,3}$ as such a $P_5$ can be detected by some node due to Lemma~\ref{lem:Fbad-path-detection} and Lemma~\ref{lem:E33-path-detection}.

For each $H_i \in \mathcal{H}$, the node $r$ counts the number of induced copies of $H_i$ that are dangerous as in Lemma~\ref{lem:5-node-graph-counting-with-bad-edges}.
Let $t(H_i)$ be the count and let $t=\sum_{i\in[16]}t(H_i)$ be the sum of them.

Finally, $r$ counts the number of $P_5$'s induced by $E\setminus (F_{bad}\cup E_{3,3})$. If it is equal to $t$, we can conclude that $G$ is $P_5$-free as removing $F_{bad}\cup E_{3,3}$ from $G$ increases the number of induced $P_5$'s by $t$. Otherwise, if the count is strictly larger than $t$, $G$ is not $P_5$-free.
\end{proof}

\subsection{Local certification of $P_5$-freeness}

We use the following standard technique.

\begin{lemma}\label{lem:certification-on-BFS-tree}
    Let $G=(V,E)$ be an $n$-node connected graph, and $T$ be a spanning tree rooted at arbitrary node $t$. For each $u\in V$, let $a(u)\in [\mathrm{poly}(n)]$ be an integer assigned to $u$. Then, there is a $\mathsf{PLS}_1$ that computes $\sum_{u\in V} a(u)$ with certificates of size $O(\log n)$.
\end{lemma}
\begin{proof}
    Let $T_u$ be a subtree of $T$ rooted at $u$.
    The certificate to $u$ consists of two $O(\log n)$-bit values $b$ and $b(u)$. If $u$ is a leaf, $u$ rejects when $a(u)\neq b(u)$.
    If $u$ is not a leaf, $u$ rejects when $b(u)\neq a(u) + \sum_{v\in N(u) \cap T_u}b(v)$. The root node $t$ rejects when $b \neq b(t)$.
    If no node rejects, $b = \sum_{u\in V} a(u)$.
\end{proof}

\begin{proof}[Proof of Theorem~\ref{thm:local-certification}]
We first describe the certificates and then explain the verification process and its correctness.

The certification comprises the following components, each represented by $O(n\log n)$ bits. We assume that $\sqrt{m/n}$ is an integer (otherwise, we use $\lceil\sqrt{m/n}\rceil$).
\begin{itemize}
    \item $\mathsf{Neighbors}$: For each node $u$, $\mathsf{Neighbors}(u)$ consists of a set of neighbors of $u$.
    \item $\mathsf{Diameter}$: $\mathsf{Diameter}(u)$ consists of a BFS tree starting from $u$.
    \item $\mathsf{SpanningTree}$: $\mathsf{SpanningTree}(u)$ describes a spanning tree rooted at the node with maximum degree, denoted $r$. Thus, it is identical to $\mathsf{Diameter}(r)$. 
    \item $\mathsf{Partition}$: $\mathsf{Partition}(u)$ provides a partition of nodes. The prover divides $V$ into $\sqrt{m/n}$ subsets $\mathcal{V} = \{V^i\}_{i\in [\sqrt{m/n}]}$ so that the number of edges between $V^i$ and $V^j$ is $O(n)$ for all $i,j\in [\sqrt{m/n}]$. This partition exists due to Lemma~\ref{lemma_partition}. $\mathsf{Partition}(u)$ includes the partition $\mathcal{V}$ and a set containing all the neighbors of $r$ in ascending order with respect to their IDs.
    Additionally, $\mathsf{Partition}(u)$ includes encoded values in $O(\log n)$-bit: $\tilde{m}_{2,2}$, $m_{bad}$, $m_{2,3}$, $m_{3,3}$, $\tilde{m}_{2,2}(u)$, $m_{bad}(u)$, $m_{2,3}(u)$, $m_{3,3}(u)$. Finally, $\mathsf{Partition}(u)$ includes, for each $H\in \mathcal{H}$, the numbers $c(H)$ and $c_u(H)$.
    \item $\mathsf{Edges}$: Let $v_k$ be the neighbor of $r$ whose degree is the $k$-th largest among all the neighbors of $r$ for $k\in [m/n]$. $\mathsf{Edges}(v_k)$ then contains $O(n\log n)$-bit descriptions of all edges between $V^i$ and $V^j$ where $k = i + (j-1)\sqrt{m/n}$.
\end{itemize}

\paragraph*{Verification}
$\mathsf{Diameter}$ is used to certificate the graph diameter is at most 3. This is done by checking the structure of the trees as in~\cite{censor2020approximate}.\vspace{2mm}

\noindent\textbf{Certifying the partition of nodes and edges.} Each node $u$ checks that the partition of nodes and the set representing all the neighbors of $r$ in ascending order with respect to their IDs written in $\mathsf{Partition}(u)$, and the spanning tree written in $\mathsf{SpanningTree}(u)$ are the same as of its neighbors, so that all nodes agree the same partition, and the distance from $r$ for all nodes. 
Now, given a partition of nodes in $\mathsf{Partition}$ and the spanning tree in $\mathsf{SpanningTree}$, endpoints of each edge $e$ determine which subset of
$$E_{1,1}\cup E_{1,2} \cup \tilde{E}_{2,2} \cup F_{bad} \cup E_{2,3} \cup E_{3,3}$$
the edge $e$ belongs to. Next, each node computes the sizes of edge sets $\abs{E_{1,1}}$, $\abs{E_{1,2}}$, $\abs{\tilde{E}_{2,2}}$, $\abs{F_{bad}}$, $\abs{E_{2,3}}$, $\abs{E_{3,3}}$ using the protocol of Lemma~\ref{lem:certification-on-BFS-tree} and the values written in $\mathsf{Partition}$ as follows: For example, we can compute $\abs{E_{2,3}}$ by setting $b=\abs{E_{2,3}}$, $b(u) = m_{2,3}(u)$, and $a(u)$ be the number of edges in $E_{2,3}$ incident to $u$.\vspace{2mm}

\noindent\textbf{Certifying the assigned edges.} Let $v_k$ be the neighbor of $r$ that has the $k$-th minimum ID among all the neighbors of $r$. Since all nodes agree on the set representing all the neighbors of $r$ in ascending order with respect to their IDs, all nodes know which is $v_k$.
For each $v_k$, the validity of $\mathsf{Edges}(v_k)$ can be checked as follows. 
First, each $v_k$ checks if the following holds.
\begin{enumerate}
    \item $\mathsf{Edges}(v_k)\subseteq E\cap (V^i\times V^j)$ for $k = i + (j-1)\sqrt{m/n}$. $v_k$ can check this since all nodes agree on the partition $\mathcal{V}$.
    \item For each edge $e\in \mathsf{Edges}(v_k)$, at least one of its endpoints is at most distance 2 from $v_k$. $v_k$ can check this since $v_k$ knows nodes with distance 2 from $v_k$ by looking at $\mathsf{Neighbors}$ of its neighbors. 
\end{enumerate}
If these conditions do not hold then $v_k$ rejects.
Each neighbor $u$ of $v_k$ rejects if $\mathsf{Edges}(v_k)$ contains non-edge that includes $u'$ for some $u'\in N(u)$ by looking at $\mathsf{Neighbors}(u')$.  This ensures that the edge set $\mathsf{Edges}(v_k)$ is indeed the set of edges between $V^i$ and $V^j$ such that at least one of the endpoints are a 2-hop neighbor of $v_k$. Each node then checks whether the four conditions described in $\mathcal{P}_{\mathsf{collect}}$ hold. If any condition holds, the network rejects. This part of verification ensures that $r$ learns $E\backslash(F_{bad}\cup E_{3,3})$.
\vspace{2mm}

\noindent\textbf{Detecting an induced $P_5$ containing $F_{bad}\cup E_{3,3}$.}
Assuming Conditions 1-4 of $\mathcal{P}_{\mathsf{collect}}$ are not met, some node can now detect an induced $P_5$ that contains at least one edge from $F_{bad}\cup E_{3,3}$ if exists, as in Lemma~\ref{lem:Fbad-path-detection} and Lemma~\ref{lem:E33-path-detection} since all necessary information from $\mathcal{P}_{\mathsf{collect}}$ is included in its certificate or the certificates of its neighbors.
\vspace{2mm}

\noindent\textbf{Detecting an induced $P_5$ solely composed of
$E\backslash(F_{bad}\cup E_{3,3})$.}
The final components of $\mathsf{Partition}(u)$ are used to count the number of dangerous $H$ for $H\in \mathcal{H}$, as described in Lemma~\ref{lem:5-node-graph-counting-with-bad-edges}. This is accomplished using Lemma~\ref{lem:certification-on-BFS-tree} as follows. For simplicity, let us consider the case of induced $C_5$ ($H_1$ in Figure~\ref{fig:5-node-graphs-all}). In our $\mathsf{CONGEST}$ algorithm, as in the proof of Lemma~\ref{lem:5-node-graph-counting-with-bad-edges}, for each such cycle, the node in $V_1$ with minimum identifier that is incident to the cycle detects it, and report the number of detected cycles to the node $r$ through the BFS tree to compute the entire number of dangerous $C_5$'s in the graph. Setting $b=c(H_1)$, $b(u)=c_u(H_1)$, and $a(u)$ be the number of dangerous $C_5$'s detected by $u$ in Lemma~\ref{lem:certification-on-BFS-tree}, we can certify it.
$r$ then determines the number of $P_5$'s that increase by removing $F_{bad}\cup E_{3,3}$ (i.e., the number of dangerous copies of $H$ for all $H\in \mathcal{H}$) and compares this with the number of induced $P_5$'s in $(V,E\backslash(F_{bad}\cup E_{3,3}))$. Based on this comparison, $r$ can decide if $G$ is $P_5$-free.

\end{proof}

\section{Lower bounds for $P_k$-freeness} \label{sec:P_5_lower_bound}

\subsection{Proof of Theorem~\ref{thm:5-path-lower-bound}}
Our proof leverages a reduction from the three-party nondeterministic communication complexity of the set-disjointness problem in the number-on-forehead (NOF) model. In this model, each player sees the inputs of the other two players but not their own. The players communicate by writing bit strings on a shared blackboard, and the communication cost of the protocol is the total number of bits written.

As a first step, we use the reduction from set-disjointness to triangle-freeness established by Drucker et al.\cite{drucker2014power}. The key graph used in their reduction is a tripartite graph $G_{RS}=(A\cup B\cup C, E)$ where (1) $\abs{A}=\abs{B}=n$ and $\abs{C}=n/3$ for arbitrary integer $n$; (2) $G_{RS}$ contains $t = \frac{n^2}{e^{O(\sqrt{\log n})}}$ triangles $\mathcal{T}=\{T_1,\ldots, T_t\}$; (3) each edge of the graph belongs to exactly one triangle. This graph construction relates to the Ruzsa–Szemerédi problem\cite{komlos1995szemeredi} in extremal graph theory and has been well studied.

For $X_A,X_B,X_C\in\{0,1\}^t$, we define $G(X_A,X_B,X_C)$ as the graph obtained by removing edges from $G_{RS}$ as follows. Let $T_i=(e_A,e_B,e_C)$ be a triangle such that $e_A$ is an edge between $B$ and $C$, $e_B$ is an edge between $A$ and $C$, and $e_C$ is an edge between $A$ and $B$. For $P\in \{A,B,C\}$, we remove the edge $e_P$ iff the $i$-th bit of $X_P$ is 0. 

We then construct another graph $G^*(X_A,X_B,X_C)$ as follows. This construction is obtained by tweaking the construction in~\cite{dalirrooyfard2019graph,nikabadi2022beyond}.
Let $\overline{G}(X_A,X_B,X_C)$ be the complement graph of $G(X_A,X_B,X_C)$.
For each node $u$ in $\overline{G}(X_A,X_B,X_C)$, we have a triangle $u_1,u_2,u_3$ in $G^*(X_A,X_B,X_C)$. We add edges $(u_i,v_i)$ for all $i\in\{1,2,3\}$. For each edge $(u,v)$ in $\overline{G}(X_A,X_B,X_C)$, we add edges $(u_i,v_j)$ for all $i,j\in\{1,2,3\}$ in $G^*(X_A,X_B,X_C)$. Additionally, there are two extra nodes $x,y$. $x$ is connected to $u_1$ and $u_2$ for all $u\in V(\overline{G}(X_A,X_B,X_C))$. $y$ is connected to $u_2$ and $u_3$ for all $u\in V(\overline{G}(X_A,X_B,X_C))$. Thus, the number of nodes in $G^*(X_A,X_B,X_C)$ is $O(n)$. 


\begin{lemma}\label{lem:NOF-disj-reduction}
 $G^{*}(X_A,X_B,X_C))$ contains an induced $P_5$ if and only if $X_A$, $X_B$ and $X_C$ are not disjoint. 
\end{lemma}
\begin{proof}
    Suppose that there is an index $i$ such that the $i$-th bit of $X_A$, $X_B$ and $X_C$ are all 1. Then we have a triangle $T_i$ in the graph $G(X_A,X_B,X_C)$ and thus have an independent set $\{u,v,w\}$ of size 3 in $\overline{G}(X_A,X_B,X_C)$. $\{u_1,x,v_2,y,w_3\}$ induces $P_5$ in $G^*(X_A,X_B,X_C)$.

    Conversely, assume that $G^{*}(X_A,X_B,X_C)$ contains an induced $P_5$. Let us focus on an independent set of size 3 in the path. The set does not contain $x$ and $y$ at the same time, since each other node is connected at least one of $x$ and $y$. Moreover, the set cannot contain $x$, as otherwise the other two nodes are written $u_3$ and $v_3$, but it is impossible since $u_3$ and $v_3$ are always connected. A similar contradiction occurs if the set contains $y$. Therefore three nodes in the set are written $u_1$, $v_2$, $w_3$. From the construction $u,v,w$ in $\overline{G}(X_A,X_B,X_C)$ is an independent set. This means that $X_A$, $X_B$ and $X_C$ are not disjoint.
\end{proof}

We are now ready to prove Theorem~\ref{thm:5-path-lower-bound}.
\begin{proof}[Proof of Theorem~\ref{thm:5-path-lower-bound}]
    Suppose that $P_5$-freeness can be certified with certificates of size $s$. We construct a non-deterministic communication protocol for set-disjointness with communication complexity $O(ns)$. Suppose that Alice, Bob, and Charlie, who receive the inputs $X_A$, $X_B$, and $X_C$, respectively, simulate $G^*(X_A,X_B,X_C)$ as follows:
\begin{itemize}
    \item For all $u\in A$ in $\overline{G}(X_A,X_B,X_C)$, Alice simulates nodes $u_1,u_2,u_3$ in $G^*(X_A,X_B,X_C)$.
    \item For all $u\in B$ in $\overline{G}(X_A,X_B,X_C)$, Bob simulates nodes $u_1,u_2,u_3$ in $G^*(X_A,X_B,X_C)$.
    \item For all $u\in C$ in $\overline{G}(X_A,X_B,X_C)$, Charlie simulates nodes $u_1,u_2,u_3$ in $G^*(X_A,X_B,X_C)$.
    \item All players simulates $x$ and $y$.
\end{itemize}
Since the incident edges of simulated nodes by each player are independent from its input, each player can construct the inputs of its simulated nodes locally. Then the players simulate the certification for $P_5$-freeness. Since the certificate size is $s$ and each player simulates $O(n)$ nodes, the length of the certificates for each player is $O(ns)$. In order to determine the outputs of the simulated nodes, each player writes its certificates to the shared blackboard. 
From Lemma~\ref{lem:NOF-disj-reduction}, the players can compute set-disjointness of size $\frac{n^2}{e^{O(\sqrt{\log n})}}$. We now obtain $s= \Omega\left(\frac{n}{e^{O(\sqrt{\log n})}}\right) = \Omega(n^{1-o(1)})$ as desired by the nondeterministic multiparty communication complexity of set-disjointness in the number-on-forehead model~\cite{rao2015simplified}.
To extend the lower bound for larger locality, we just replace edges incident to $x$ and $y$ by longer paths. Thus we can increase the locality by 1 at the cost of increasing the path length by 4.
\end{proof}

\subsection{Lower bounds for $P_k$-freeness: general $k$}\label{sec:P11-freeness}

\paragraph*{Technical challenges and the general proof idea for Theorem~\ref{thm:11-path-lower-bound}}
We use a standard reduction from two-party set-disjointness function for the proof of Theorem~\ref{thm:11-path-lower-bound} (and Theorem~\ref{thm:local-certification-lower-bound}). 
Alice and Bob jointly construct some graph which depends on their inputs so that the constructed graph is $P_k$-free iff the output of some function (in our case, the set-disjointness function) is 1. 

The main challenge in this proof is to ensure that the cut of the constructed graph (i.e., the number of edges between nodes held by Alice and those held by Bob) is as sparse as possible. This was not an issue in the proof of Theorem~\ref{thm:5-path-lower-bound}, which uses the non-deterministic communication complexity of the NOF model or in the proof of the lower bound given in~\cite{bousquet2024local}, which uses counting arguments. In fact, the constructions used in these proofs do not achieve a sparse cut, making them unsuitable for the proof of Theorem~\ref{thm:11-path-lower-bound} that holds for the $\mathsf{CONGEST}$ model. 

Fortunately, in the simplest case where $d\geq 3$, we can utilize the construction from Ref.~\cite{legallISAAC21} with a slight modification, taking advantage of the graph structure that holds for all $d\geq 3$. However, for the cases where $d= 1$ and $d=2$, the construction from Ref.~\cite{legallISAAC21}, which is applicable to induced $k$-cycles for every $k\geq 4$, fails for induced $k$-paths. This is intuitively because we lose several symmetric properties of induce $k$-cycles by removing one edge from the cycle. We thus construct two other graph families from scratch. The main challenge in establishing these constructions is that both edges and non-edges must be carefully chosen to avoid creating unexpected induced $k$-paths. Our proof requires a detailed analysis to ensure this, as the general structure valid for $d\geq 3$ does not apply for $d\leq 2$.

\subsubsection{Lower bound for $d = 1$}

\paragraph*{The graph construction.}
We assume that $x$ and $y$ are bit strings of length $K=n^2$ for some integer $n$. Here we construct a specific graph $G_{x,y}=(V,E_{x,y})$ as follows. 
The node set $V$ of $G_{x,y}$ is the union of two subsets $V_A$ and $V_B$.
$V_A$ consists of three node sets $A_1,A_2$ of size $n$ and four specific nodes $p,q,r,s$. $V_B$ consists of two node sets $B_1,B_2$ of size $n$ and three specific nodes $x,y,z$. Each set is explicitly described as $A_i = \{a_{i,j}|j\in [n]\}$ and $B_i = \{b_{i,j}|j\in [n]\}$ for $i\in \{1,2\}$.

The following edges are contained for all $G_{x,y}$, regardless of the value of $x$ and $y$. See Figure~\ref{fig:graph-P11} for the illustration.
\begin{itemize}
    \item $A_1,A_2,B_1,B_2$ induce $n$-node cliques, respectively;
    \item $(a_{1,j},b_{1,j})$, $(a_{2,j},b_{2,j})$ for all $j\in [n]$;
    \item $q$ is connected all nodes in $A_1$;
    \item $r$ is connected all nodes in $A_2$;
    \item $x$ is connected all nodes in $B_1$;
    \item $z$ is connected all nodes in $B_2$;
    \item $p$ is connected to $q$, and $s$ is connected to $r$;
    \item there is a path $x,y,z$.
\end{itemize}
Furthermore, we add the following edges depending on the value of $x$ and $y$. 
For each $j_1,j_2\in [n]$, let $k = j_1 + (j_2-1)\cdot (n-1)$.  
\begin{itemize}
    \item If $x_k = 0$, we add the edge $(a_{1,j_1},a_{2,j_2})$.
    \item If $y_k = 0$, we add the edge $(b_{1,j_1},b_{2,j_2})$.
\end{itemize}
This finishes the construction of $G_{x,y}$.

\begin{figure}[tbp]
    \begin{center}
        \includegraphics[scale=1]{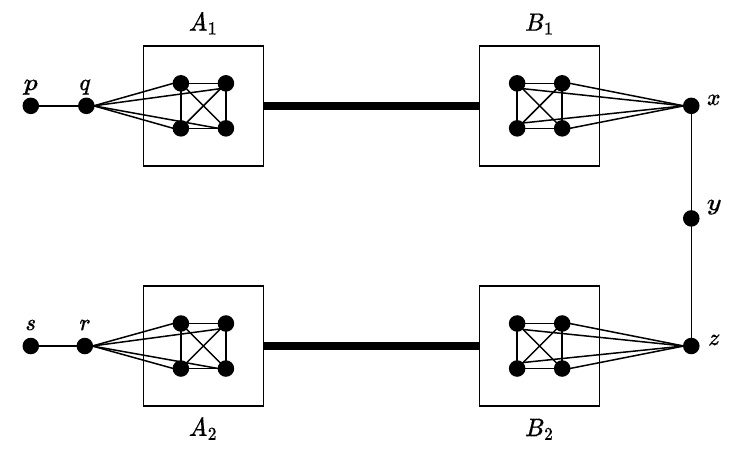}
    \end{center} 
\caption{An illustration of the fixed graph construction. Thick edges between $A_1$ and $B_1$, and $A_2$ and $B_2$ represents a perfect matching ($n$ edges between a pair of nodes with the same label).}
\label{fig:graph-P11}
\end{figure}

In the constructed graph $G_{x,y}$, any induced $P_{11}$ is forced to use the nodes $p,q,r,s,x,y,z$. Therefore the remaining four nodes should be $A_1$-node, $A_2$-node, $B_1$-node, and $B_2$-node. We also added shortcut edges between $A_1$ and $A_2$ and between $B_1$ and $B_2$ so that there is no shortcut edge in the path iff $\mathsf{DISJ}_{n^2}(x,y) = 0$. We can justify this intuition as the following lemma. 
\begin{lemma}\label{lem:11-path-graph-construction}
    $G_{x,y}$ is $P_{11}$-free iff $\mathsf{DISJ}_{n^2}(x,y) = 1$.
\end{lemma}
\begin{proof}
    We try to exploit an induced 11-path from the graph $G_{x,y}$. From the following analysis, we can see that any induced 11-path in $G_{x,y}$ contains at most one $A_1$ node and one $A_2$ node.
    \begin{itemize}
        \item \textbf{A path with at least two $A_1$ nodes.} Let $a_1$ and $a'_1$ be two such nodes. We cannot use $p,q$ since $\{a_1,a'_1,q\}$ forms a triangle. Suppose that $a_1$ is an end point of the path. Then the path can be extended at most 10 nodes as follows:
        $$ a_{1}-a'_{1}- (B_1)-x-y-z-(B_2)-(A_2)-r-s$$
        Here $(B_1)$ ,$(B_2)$, and $(A_2)$ means arbitrary node from the node set $B_1$, $B_2$, and $A_2$, respectively.
        Suppose that both $a_1$ and $a'_1$ are intermediate points of the path. Then without loss of generality, assume that the neighbors of $a_1$ is $B_1$-node and of $a'_1$ is $A_2$-node.
        On the $B_1$ side, the path will be extended as follows.
        \begin{align*}
            a_1 \rightarrow (B_1) \rightarrow
            \left\{
            \begin{array}{l}
             (B_1) \rightarrow (B_2)\\
             x \rightarrow y \rightarrow z \rightarrow (B_2)\\
             (B_2) \rightarrow (B_2) \\
             (B_2) \rightarrow z \rightarrow y
            \end{array}
            \right.
        \end{align*}
        On the $A_2$ side, the path will be extended as follows.
        \begin{align*}
            a'_1 \rightarrow (A_2) \rightarrow
            \left\{
            \begin{array}{l}
             r \rightarrow s \\
             (A_2) \rightarrow (B_2) \rightarrow (B_2) \\
             (A_2) \rightarrow (B_2) \rightarrow z \rightarrow y
            \end{array}
            \right.
        \end{align*}
        The path can be extended at most 10 nodes as follows:
        \begin{align*}
        &a_{1}-(A_{2})-r-s \\
        &\mid \\
        &a'_{1}-(B_1)-x-y-z-(B_2)
        \end{align*}
        \item \textbf{A path with at least two $A_2$ nodes.} From the symmetry of the graph, it is impossible to extract an induced 11-path with two $A_2$ nodes. 
    \end{itemize}
    Since $B_1$ and $B_2$ are cliques, we can use at most two nodes from these sets. Therefore we are able to choose at most $7$ nodes from $V_B= B_1\cup B_2 \cup \{x,y,z\}$. We thus need to choose at least 4 nodes from $V_A= A_1\cup A_2 \cup \{p,q,r,s\}$. Therefore exactly one node from $A_1$ and one node from $A_2$ should be contained in the path since $\{p,q\}$ is reachable from $V_B$ only through $A_1$ nodes, and $\{r,s\}$ is reachable from $V_B$ only through $A_2$ nodes. We write such nodes $a_1 \in A_1$ and $a_2 \in A_2$. Observe that it is impossible that both $a_1$ and $a_2$ are the endpoints of an induced 11-path. We next analyze the following two cases.
    \begin{itemize}
        \item \textbf{A path with an endpoint $a_1 \in A_1$ and an intermediate node $a_2 \in A_2$.} 
        If $a_1$ is connected to a $B_1$-node, the path will be extended as follows (including the case the path does not contain $a_2 \in A_2$):
        \begin{align*}
            a_1 \rightarrow (B_1) \rightarrow
            \left\{
            \begin{array}{l}
             x \rightarrow y \rightarrow z \rightarrow (B_2) \rightarrow a_2 \rightarrow r \rightarrow s \\
             (B_1) \rightarrow (B_2) \rightarrow 
             \left\{
             \begin{array}{l}
                z \rightarrow y \\
                a_2 \rightarrow r \rightarrow s
             \end{array}\right. \\
             (B_2) \rightarrow \left\{
             \begin{array}{l}
                z \rightarrow y \\
                a_2 \rightarrow r \rightarrow s
             \end{array}\right. 
            \end{array}
            \right.
        \end{align*}
        If $a_1$ is connected to $a_2$, the path will be extended as follows:
        \begin{align*}
            a_1 \rightarrow a_2 \rightarrow
            \left\{
            \begin{array}{l}
             (B_2) \rightarrow 
             \left\{\begin{array}{l}
                x \rightarrow y \rightarrow z \rightarrow (B_1)\\ 
                (B_2) \rightarrow (B_1) \rightarrow x \rightarrow y\\
                (B_1) \rightarrow x \rightarrow y
            \end{array}\right. \\
             r \rightarrow s 
            \end{array}
            \right.
        \end{align*} Therefore, there is no induced 11-path containing an endpoint $a_1 \in A_1$ and an intermediate node $a_2 \in A_2$.
        \item \textbf{A path with two intermediate nodes $a_1 \in A_1$ and $a_2 \in A_2$.} 
        If $a_1$ is connected to $q$, the path will be extended as follows: 
        \begin{align*}
            p \rightarrow q \rightarrow a_1 \rightarrow
            \left\{
            \begin{array}{l}
             (B_1) \rightarrow 
             \left\{\begin{array}{l}
                    x \rightarrow y \rightarrow z \rightarrow (B_2) \rightarrow a_2 \rightarrow r \rightarrow s\\ 
                    (B_2) \rightarrow \left\{\begin{array}{l}
                            z \rightarrow y \\
                            (B_2) \rightarrow a_2 \rightarrow r \rightarrow s\\
                            a_2 \rightarrow r \rightarrow s
                        \end{array}\right. \\
                    (B_1) \rightarrow (B_2) \rightarrow
                    \left\{\begin{array}{l}
                            (B_2) \rightarrow a_2 \rightarrow r \rightarrow s \\
                            a_2 \rightarrow r \rightarrow s
                        \end{array}\right.
                \end{array}\right. \\
             a_2 \rightarrow \left\{\begin{array}{l}
                    r \rightarrow s \\
                    (B_2) \rightarrow \left\{\begin{array}{l}
                        (B_2) \rightarrow (B_1) \rightarrow 
                        \left\{\begin{array}{l}
                            (B_1) \\ 
                            x \rightarrow y 
                        \end{array}\right. \\
                        (B_1) \rightarrow 
                        \left\{\begin{array}{l}
                            (B_1) \\ 
                            x \rightarrow y
                        \end{array}\right.
                    \end{array}\right. \\
                \end{array}\right. \\
            \end{array}
            \right.
        \end{align*}
    \end{itemize}
    We conclude that the only induced $P_{11}$ we are able to extract from the graph $G_{x,y}$ will be of the form
        $$
         p \rightarrow q \rightarrow a_1 \rightarrow (B_1) \rightarrow x \rightarrow y \rightarrow z \rightarrow (B_2) \rightarrow a_2 \rightarrow r \rightarrow s.
        $$
        Suppose that 
        $$
         p \rightarrow q \rightarrow a_{1,i} \rightarrow b_{1,k} \rightarrow x \rightarrow y \rightarrow z \rightarrow b_{2,j} \rightarrow a_{2,\ell} \rightarrow r \rightarrow s
        $$
        is an induced $P_{11}$ for some $i,j,k,\ell \in [n]$. From the construction of $G_{x,y}$ we have $i = k$  otherwise $a_{1,i}$ is not a neighbor of $b_{1,k}$. Similarly, we have $j=\ell$.
        Finally, we can observe that there is no edge between $a_{1,i}$ and $a_{2,j}$ and between $b_{1,i}$ and $b_{2,j}$ iff $x_{k}= y_{k} = 1$ for $k = i + (j-1)(n-1)$. Now the proof is completed.
\end{proof}

\noindent\textbf{Proof of Theorem~\ref{thm:11-path-lower-bound} for $d=1$:}
We utilize the standard framework of reductions from two-party communication complexity. Consider the two-party set-disjointness function $\mathsf{DISJ}_K:\{0,1\}^K\times\{0,1\}^K\rightarrow \{0,1\}$: for $x,y\in \{0,1\}^K$, $\mathsf{DISJ}_K(x,y) = 0$ if and only if there exists an index $i\in [K]$ such that $x_i=y_i=1$ where $x_i$ and $y_i$ are the $i$-th bit of $x$ and $y$, respectively. It is well known that if Alice has $x$ and Bob has $y$ as their inputs, the total amount of bits communicated between them in any two-party communication protocol that computes $\mathsf{DISJ}_K(x,y)$ with probability at least $2/3$ is $\Omega(K)$ bits~\cite{razborov1990distributional}. 


Let $\mathcal{A}$ be an $r$-round algorithm that solves $P_{11}$-freeness. Assume that Alice simulates 
$$
V_A = A_1\cup A_2 \cup \{p,q,r,s\}
$$ and Bob simulates$$
V_B = B_1\cup B_2 \cup \{x,y,z\}.
$$
They can simulate 1-round of the algorithm $\mathcal{A}$ using $O(n\log n)$ bits of communication since there are $2n$ edges between $V_A$ and $V_B$. Therefore they can decide if $G_{x,y}$ is $P_{11}$-free using $O(rn\log n)$ bits of communication.
By Lemma~\ref{lem:11-path-graph-construction} it means they can compute the value $\mathsf{DISJ}_{n^2}(x,y)$ as well. Since the two-party communication complexity of $\mathsf{DISJ}_{n^2}$ is $\Omega(n^2)$~\cite{razborov1990distributional}, it follows that $r=\Omega(n/\log n)$, which completes the proof of Theorem~\ref{thm:11-path-lower-bound} for $d=1$.

\subsubsection{Lower bound for $d\geq 3$}

For $d\geq 3$, our graph construction is almost the same as of Ref.~\cite{legallISAAC21} for induced cycles. Below we repeat the construction for completeness.

\noindent\textbf{Vertices.}
We define the sets of vertices as follows:

\begin{itemize}
    \item $A_1=A_1^1\cup\cdots \cup A_1^{n}$, where
    $A_1^i=\left\{ a_1^{i,j} \middle| 0\leq j \leq d-1 \right\}$ for $i\in [n]$.
    \item $A_2=A_2^1\cup\cdots  \cup A_2^{n}$, where
    $A_2^i=\left\{ a_2^{i,j} \middle| 0\leq j \leq d-1 \right\}$ for $i\in [n]$.
    \item $B_1=B_1^1\cup\cdots\cup B_1^{n}$, where
    $B_1^i=\left\{ b_1^{i,j} \middle| 0\leq j \leq d-1 \right\}$ for $i\in [n]$.
    \item $B_2=B_2^1\cup\cdots\cup B_2^{n}$, where
    $B_2^i=\left\{ b_2^{i,j} \middle| 0\leq j \leq d-1 \right\}$ for $i\in [n]$.
    \item
    $U_A=\left\{ u_A^{i} \middle| 0\leq i \leq d n^{1/d} \right\}$.
    \item 
    $L_A=\left\{ l_A^{i} \middle| 0\leq i \leq d n^{1/d} \right\}$.
    \item 
    $U_B=\left\{ u_B^{i} \middle| 0\leq i \leq d n^{1/d} \right\}$.
    \item
    $L_B=\left\{ l_B^{i} \middle| 0\leq i \leq d n^{1/d} \right\}$.
\end{itemize}

\noindent Each $S\in \{A_1,A_2,B_1,B_2\}$ contains $d n$ vertices, and divided into $n$ subsets of size $d$.
Each $C\in \{U_A,U_B,L_A,L_B\}$ contains $d n^{1/d}$ vertices. The number of vertices $V=A_1\cup A_2\cup B_1\cup B_2 \cup U_A \cup L_A \cup U_B \cup L_B$ is $\Theta(\ell n + d n^{1/d}) = \Theta(n)$.\\

\noindent\textbf{Edges.}
First, we add $2d n^{1/d}$ edges $\left\{(u_A^i,u_B^i),(l_A^i,l_B^i)\middle| i\in [d n^{1/d}]\right\}$.
Then, we consider a map from $[n]$ to $[d n^{1/\ell}]^d$, where $[d n^{1/d}]^d$ is $d$ times direct product of the set $[d n^{1/d}]$.
Since
\begin{align*}
    \begin{pmatrix}
    d n^{1/d}\\
    d
    \end{pmatrix}
    = \frac{d n^{1/d}}{d} \cdot \frac{d n^{1/d}-1}{d-1} \cdots \frac{d n^{1/d}-d + 1}{1}
    \geq \left(\frac{d n^{1/d}}{d}\right)^{d} = n
\end{align*}
holds, there exists an injection $\sigma:[n]\rightarrow [d n^{1/d}]^{d}$.
We arbitrarily choose one of these injections.
For $i\in [n]$, we denote $\sigma(i)=\{k_1,\ldots ,k_{d}\}\in [d n^{1/d}]^{d}$.
For all $i \in [n],j\in [d]$, we add the following edges. 
\begin{itemize}
    \item We add an edge set $\left\{ (a_1^{i,j},u_A^{k_j}) \middle| i\in [n],j\in [d] \right\}$.
    
    \item We add an edge set $\left\{ (a_2^{i,j},l_A^{k_j}) \middle| i\in [n],j\in [d] \right\}$.
    \item We add an edge set $\left\{ (b_1^{i,j},u_B^{k_j}) \middle| i\in [n],j\in [d] \right\}$.
    \item We add an edge set $\left\{ (b_2^{i,j},l_B^{k_j}) \middle| i\in [n],j\in [d] \right\}$.
\end{itemize}
Now we can determine exactly $\ell$ vertices of $U_A$ that are adjacent to vertices of $A_1^i$. We denote them $Code(A_1^i)\subseteq U_A$. In the same way, we determine the vertex sets $Code(A_2^i)\subseteq L_A$, $Code(B_1^i)\subseteq U_B$, and $Code(B_2^i)\subseteq L_B$ by using the same $\sigma$.
Since $\sigma$ is an injection, it holds that $Code(A_1^i)\neq Code(A_1^j)$, $Code(A_2^i)\neq Code(A_2^j)$, $Code(B_1^i)\neq Code(B_1^j)$, and $Code(B_2^i)\neq Code(B_2^j)$ for $i\neq j$. 

\noindent In addition, we add the following edges.
\begin{itemize}
    \item For any $i,j\in [n]$, add edges between $u\in A_1^i,v\in A_1^j$ if and only if $i\neq j$.
    \item For any $i,j\in [n]$, add edges between $u\in A_2^i,v\in A_2^j$ if and only if $i\neq j$.
    \item For any $i,j\in [n]$, add edges between $u\in B_1^i,v\in B_1^j$ if and only if $i\neq j$.
    \item For any $i,j\in [n]$, add edges between $u\in B_2^i,v\in B_2^j$ if and only if $i\neq j$.
\end{itemize}

\paragraph*{Creating $G_{x,y}$}

Given two binary strings $x,y\in\{0,1\}^{n^2}$, we add the following edges:
\begin{itemize}
    \item For $i,j\in [n]$, add edges $\{(a_1^{i,k+1},a_2^{j,k}) | k\in [d-1]\}$, if and only if $x_{i,j}=1$.
    \item For $i,j\in [n]$, add edges $\{(b_1^{i,k},b_2^{j,k}) | k\in [d-1]\}$, if and only if $y_{i,j}=1$.
\end{itemize}

This concludes the description of $G_{x,y}$. The following lemmas are shown as exactly in Ref.~\cite{legallISAAC21}.

\begin{lemma}\label{C8l_lemma}
Any subset of nodes $\mathcal{C}\subseteq V$ of size $8d$ in $G_{x,y}$ which induces $P_{8d}$ contains at most $d$ nodes in $S$ where $S\in \{A_1,A_2,B_1,B_2\}$.
\end{lemma}

\begin{lemma}\label{C8l_lemma2}
Any subset of nodes $\mathcal{C}\subseteq V$ of size $8d$ in $G_{x,y}$ which induces $P_{8d}$ contains  $d$ nodes in $S$, where $S\in \{A_1,A_2,B_1,B_2,U_A,U_B,L_A,L_B\}$.
\end{lemma}

Note that Lemma~\ref{C8l_lemma} holds for $d \leq 2$ in the case of $C_{8d}$, but does not hold  for $d \leq 2$ in the case of $P_{8d}$. For example, for $d = 2$ we can see that there is an induced $P_{8d + 1}$ containing  $d + 1$ nodes in $A_1$. Now it is easy to see that $G_{x,y}$ is $P_{8d}$-free iff $x$ and $y$ are disjoint.

\subsubsection{Lower bound for $d=2$}\label{appendix-P22}

We first explain our graph construction.

\noindent\textbf{Vertices.}
We define the sets of vertices as follows:

\begin{itemize}
    \item $A_1=A_1^1\cup\cdots \cup A_1^{n}$, where
    $A_1^i=\left\{ a_1^{i,j} \middle|  j \in \{1,2\} \right\}$ for $i\in [n]$;
    \item $A_2=A_2^1\cup\cdots  \cup A_2^{n}$, where
    $A_2^i=\left\{ a_2^{i,j} \middle| j \in \{1,2\} \right\}$ for $i\in [n]$;
    \item $B_1=B_1^1\cup\cdots\cup B_1^{n}$, where
    $B_1^i=\left\{ b_1^{i,j} \middle| j \in \{1,2\} \right\}$ for $i\in [n]$;
    \item $B_2=B_2^1\cup\cdots\cup B_2^{n}$, where
    $B_2^i=\left\{ b_2^{i,j} \middle| j \in \{1,2\} \right\}$ for $i\in [n]$;
    \item
    $U_A = U_{A,1}\cup U_{A,2}$ where $U_{A,j}=\left\{ u_A^{i,j} \middle| i \in [2 n^{1/2}]  \right\}$ for $j\in\{1,2\}$;
    \item
    $L_A = L_{A,1}\cup L_{A,2}$ where $L_{A,j}=\left\{ l_A^{i,j} \middle| i \in [2 n^{1/2}]  \right\}$ for $j\in\{1,2\}$;
    \item
    $U_B = U_{B,1}\cup U_{B,2}$ where $U_{B,j}=\left\{ u_B^{i,j} \middle| i \in [2 n^{1/2}]  \right\}$ for $j\in\{1,2\}$;
    \item
    $L_B = L_{B,1}\cup L_{B,2}$ where $L_{B,j}=\left\{ l_B^{i,j} \middle| i \in [2 n^{1/2}]  \right\}$ for $j\in\{1,2\}$;
    \item $\{p,q,r,s,x,y\}$;
\end{itemize}

\noindent\textbf{Edges.}
We have edges $(p,q)$, $(r,s)$.
For each $i\in [n]$, we have edges $(q,a_1^{i,1})$, $(r,a_2^{i,1})$, $(x,b_1^{i,1})$, $(x,b_2^{i,1})$, $(y,b_1^{i,2})$, $(y,b_2^{i,2})$. 
We add $4 n^{1/2}$ edges $\left\{(u_A^i,u_B^i),(l_A^i,l_B^i)\middle| i\in [2 n^{1/2}]\right\}$.
Then, we consider a map from $[n]$ to $[2 n^{1/2}]^2$, where $[2 n^{1/2}]^2$ is $2$ times direct product of the set $[2 n^{1/2}]$.

For each $i\in [n]$, we write $i = i_1 +  (i_2-1)n^{1/{2}}$ where $i_1,i_2 \in [n^{1/2}]$. 
We add edges $(a_{1}^{i,1},u_{A}^{i_1,1})$, $(a_{1}^{i,2},u_{A}^{i_2,2})$, $(a_{2}^{i,1},l_{A}^{i_1,1})$, $(a_{2}^{i,2},l_{A}^{i_2,2})$, $(b_{1}^{i,1},u_{B}^{i_1,1})$, $(b_{1}^{i,2},u_{B}^{i_2,2})$, $(b_{1}^{i,1},u_{B}^{i_1,1})$, $(b_{1}^{i,2},l_{B}^{i_2,2})$.

Now we can determine exactly $2$ vertices of $U_A$ that are adjacent to vertices of $A_1^i$. We denote them $Code(A_1^i)\subseteq U_A$. In the same way, we determine the vertex sets $Code(A_2^i)\subseteq L_A$, $Code(B_1^i)\subseteq U_B$, and $Code(B_2^i)\subseteq L_B$.
It holds that $Code(A_1^i)\neq Code(A_1^j)$, $Code(A_2^i)\neq Code(A_2^j)$, $Code(B_1^i)\neq Code(B_1^j)$, and $Code(B_2^i)\neq Code(B_2^j)$ for $i\neq j$. 

\noindent In addition, we add the following edges.
\begin{itemize}
    \item For any $i,j\in [n]$, add edges between $u\in A_1^i,v\in A_1^j$ if and only if $i\neq j$.
    \item For any $i,j\in [n]$, add edges between $u\in B_2^i,v\in B_2^j$ if and only if $i\neq j$.
\end{itemize}
If $d \geq 2$, we add the following edges.
\begin{itemize}
    \item For any $i,j\in [n]$, add edges between $u\in A_2^i,v\in A_2^j$ if and only if $i\neq j$.
    \item For any $i,j\in [n]$, add edges between $u\in B_1^i,v\in B_1^j$ if and only if $i\neq j$.
\end{itemize}

\paragraph*{Creating $G_{x,y}$}

Given two binary strings $x,y\in\{0,1\}^{n^2}$, we add the following edges:
\begin{itemize}
    \item For $i,j\in [n]$, add an edge $(a_1^{i,2},a_2^{j,1})$, if and only if $x_{i,j}=1$.
    \item For $i,j\in [n]$, add edges $\{(b_1^{i,k},b_2^{j,k}) | k\in \{1,2\}\}$, if and only if $y_{i,j}=1$.
\end{itemize}

This concludes the description of $G_{x,y}$.

\begin{lemma}\label{P22_lemma1}
Any subset of nodes $\mathcal{P}\subseteq V$ of size $22$ in $G_{x,y}$ which induces $P_{22}$ contains at most $2$ nodes in $S$, where $S\in \{A_1,A_2,B_1,B_2\}$.
\end{lemma}
\begin{proof}
    Consider the case $A_1$ (and $A_2$). If $\abs{\mathcal{P}\cap A_1}\geq 4$, it induces $C_4$ or $C_3$. So assume that $\abs{\mathcal{P}\cap A_1} = 3$. In this case we have $\mathcal{P}\cap A_1 =\{a_1^{i,1},a_1^{i,2},a_1^{j,k}\}$ for some $i\neq j\in [n]$, $k\in\{1,2\}$, and $\mathcal{P}$ has two edges $(a_1^{i,1},a_1^{j,k}),(a_1^{i,2},a_1^{j,k})$. There are two cases as follows. See Figure~\ref{fig:graph-P20-P21} for the illustration.
    \begin{itemize}
        \item Suppose that $a_1^{i,1}$ is the endpoint in the path. In this case, we can see that the path can be extended to 20 nodes by a similar analysis as in the proof of Lemma~\ref{lem:11-path-graph-construction}.
        \item Suppose that $a_1^{i,1}$ and $a_1^{i,2}$ are not the endpoints in the path. In this case, we can see that the path can be extended to 21 nodes by a similar analysis as in the proof of Lemma~\ref{lem:11-path-graph-construction}.
    \end{itemize}

    Next, consider the case $B_1$ (and $B_2$). In this case we have $\mathcal{P}\cap B_1 =\{b_1^{i,1},b_1^{i,2},b_1^{j,k}\}$ for some $i\neq j\in [n]$, $k\in\{1,2\}$, and $\mathcal{P}$ has two edges $(b_1^{i,1},b_1^{j,k}),(b_1^{i,2},b_1^{j,k})$. There are two cases as follows. See Figure~\ref{fig:graph-P14-P19} for the illustration.
    \begin{itemize}
        \item Suppose that $b_1^{i,1}$ is the endpoint in the path. In this case, we can see that the path can be extended to 14 nodes by a similar analysis as in the proof of Lemma~\ref{lem:11-path-graph-construction}.
        \item Suppose that $b_1^{i,1}$ and $b_1^{i,2}$ are not the endpoints in the path. In this case, we can see that the path can be extended to 19 nodes by a similar analysis as in the proof of Lemma~\ref{lem:11-path-graph-construction}.
    \end{itemize}
\end{proof}

\begin{lemma}\label{P22_lemma2}
Any subset of nodes $\mathcal{P}\subseteq V$ of size $22$ in $G_{x,y}$ which induces $P_{22}$ contains exactly $2$ nodes in $S$, where $S\in \{A_1,A_2,B_1,B_2,U_A,U_B,L_A,L_B\}$.
\end{lemma}
\begin{proof}
For $S\subseteq V$, we denote $z(S)=|\mathcal{P}\cap S|$. We further define $z_{end}(S)$ by the number of nodes in $\mathcal{P}\cap S$ that are the endpoints in $\mathcal{P}$. Since $\mathcal{P}$ has two endpoints, we have $z_{end}(S)\leq 2$ for any $S$.
Observe that $z(U_A)\leq z(A_1) + z_{end}(U_A)$ since each node in $U_A$ has exactly one neighbor out of $A_1$. Similar observations show that $z(U_B)\leq z(B_1) + z_{end}(U_B)$, $z(L_A)\leq z(A_2) + z_{end}(L_A)$, and $z(L_B)\leq z(B_2) + z_{end}(L_B)$. Therefore,
\begin{align*}
    &z(U_A)+z(L_A)+z(U_B)+z(L_B)\\
    &\leq z(A_1)+z(A_2) + z(B_1)+z(B_2) + z_{end}(U_A) + z_{end}(L_A)+ z_{end}(U_B)+ z_{end}(L_B)\\
    &= z(A_1)+z(A_2) + z(B_1)+z(B_2) + z_{end}(U_A\cup L_A \cup U_B\cup L_B).
\end{align*}
Combining it with the statement of Lemma~\ref{P22_lemma1}, we have
\begin{align*}
    22 &= z(A_1)+z(A_2)+z(U_A)+z(L_A)+z(B_1)+z(B_2)+z(U_B)+z(L_B) + z(\{p,q,r,s,x,y\})\\
    &\leq 2(z(A_1)+z(A_2) + z(B_1)+z(B_2))  + z(\{p,q,r,s,x,y\}) + z_{end}(U_A\cup L_A \cup U_B\cup L_B)\\
    &\leq 16 + z(\{p,q,r,s,x,y\}) + z_{end}(U_A\cup L_A \cup U_B\cup L_B).
\end{align*}

If $z_{end}(U_A\cup L_A \cup U_B\cup L_B) = 0$, then $z(\{p,q,r,s,x,y\})=6$ and $z(A_1) = z(A_2) = z(B_1) = z(B_2) = 2$. This means that the two endpoints of $\mathcal{P}$ are $p$ and $s$, so we have $z(U_A) = z(U_B) = z(L_A) = z(L_B) = 2$.

If $z_{end}(U_A\cup L_A \cup U_B\cup L_B) = 1$, then $z(\{p,q,r,s,x,y\}) = 5$ since in this case both $p,q$ cannot be non-endpoints. Therefore $z(A_1) = z(A_2) = z(B_1) = z(B_2) = 2$. We now have $z(U_A)+z(L_A)+z(U_B)+z(L_B) = 9$. In case $z(U_A)=3$, we have $z_{end}(U_A)=1$. This endpoint cannot have any neighbor since the other two nodes in $\mathcal{P}\cap U_A$ must have two edges to $A_1$ and to $U_B$. This is a contradiction. The same contradiction occurs in the case of $z(U_B)=3$, $z(L_A)=3$, and $z(L_B)=3$. Furthermore, the similar contradiction occurs when $z_{end}(U_A\cup L_A \cup U_B\cup L_B) = 2$.
\end{proof}

From Lemma~\ref{P22_lemma2}, any induced $P_{22}$ uses exactly two nodes from each subset $S$ where $$S\in \{A_1,A_2,B_1,B_2,U_A,U_B,L_A,L_B\},$$ and nodes $\{p,q,r,s,x,y\}$. It is easy to see that $G_{x,y}$ is $P_{22}$-free iff $x$ and $y$ are disjoint.

\begin{figure}[tbp]
    \begin{center}
        \includegraphics[scale=0.7]{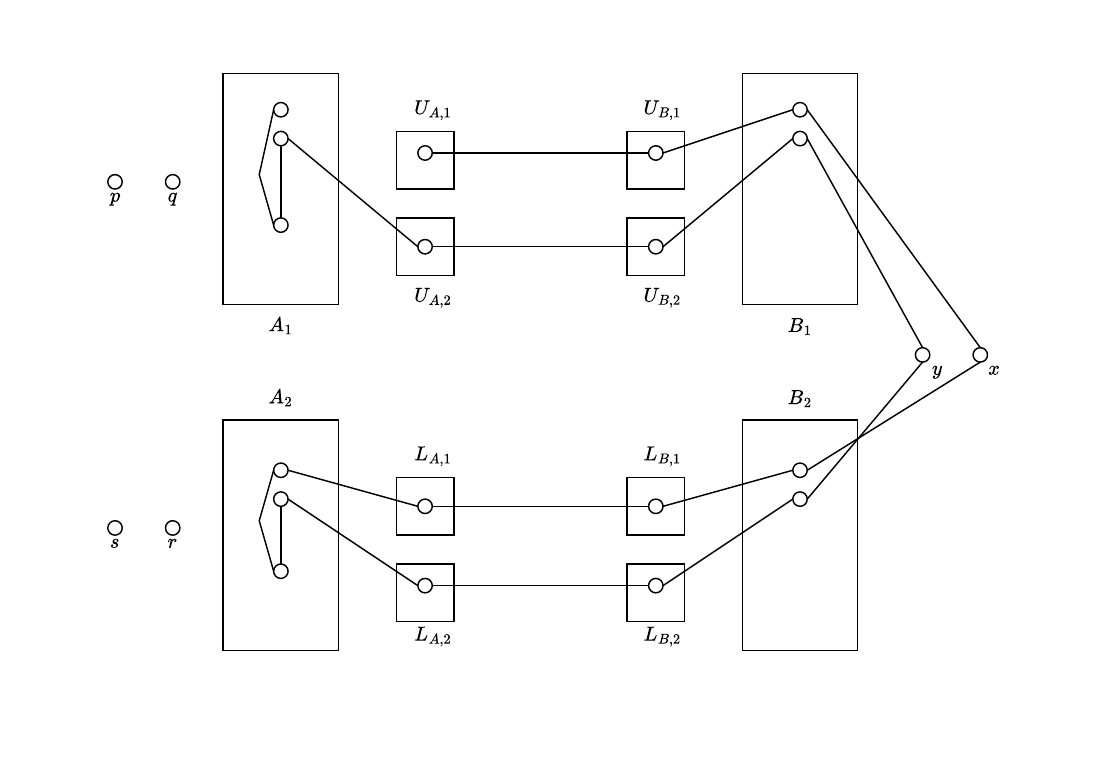}
        \includegraphics[scale=0.7]{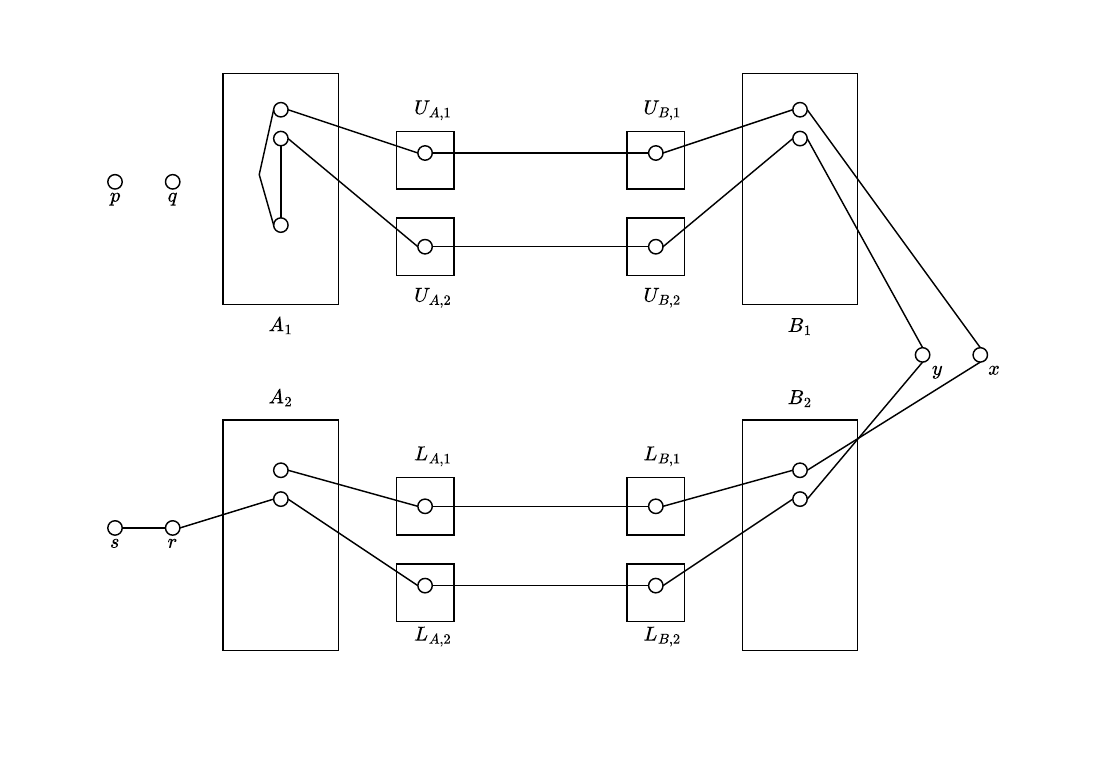}
    \end{center} 
\caption{An example of $P_{20}$ and $P_{21}$ in the proof of Lemma~\ref{P22_lemma1}.}
\label{fig:graph-P20-P21}
\end{figure}

\begin{figure}[tbp]
    \begin{center}
        \includegraphics[scale=0.7]{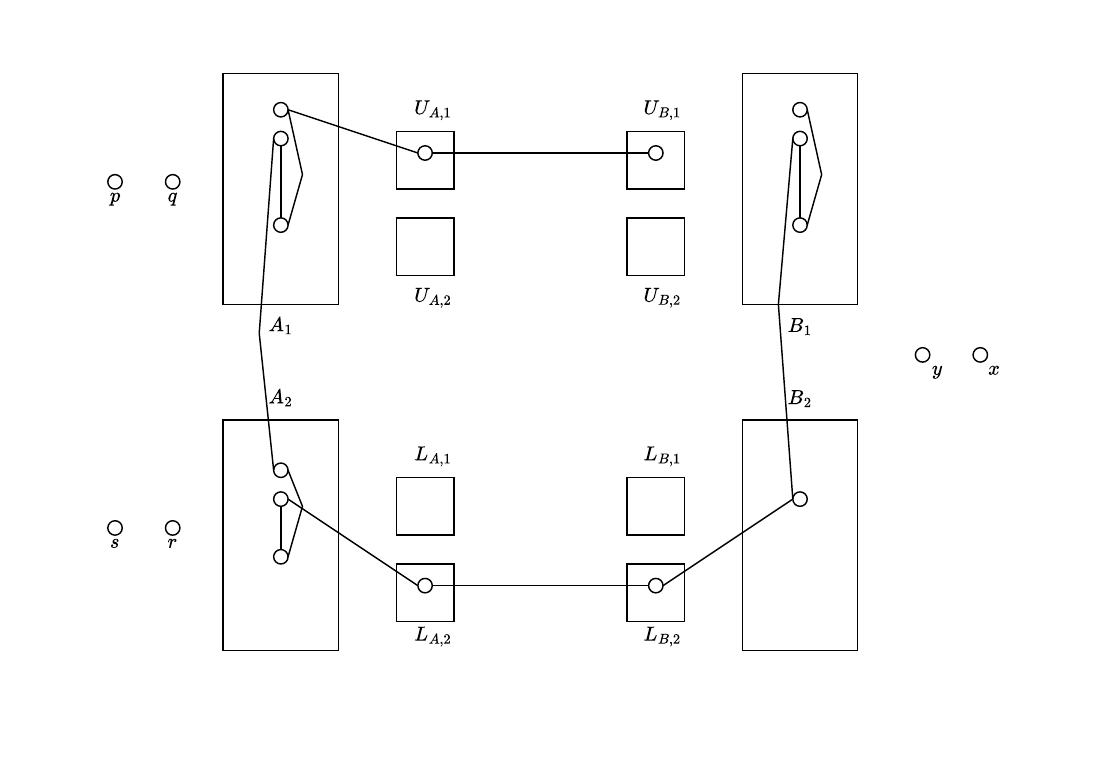}
        \includegraphics[scale=0.7]{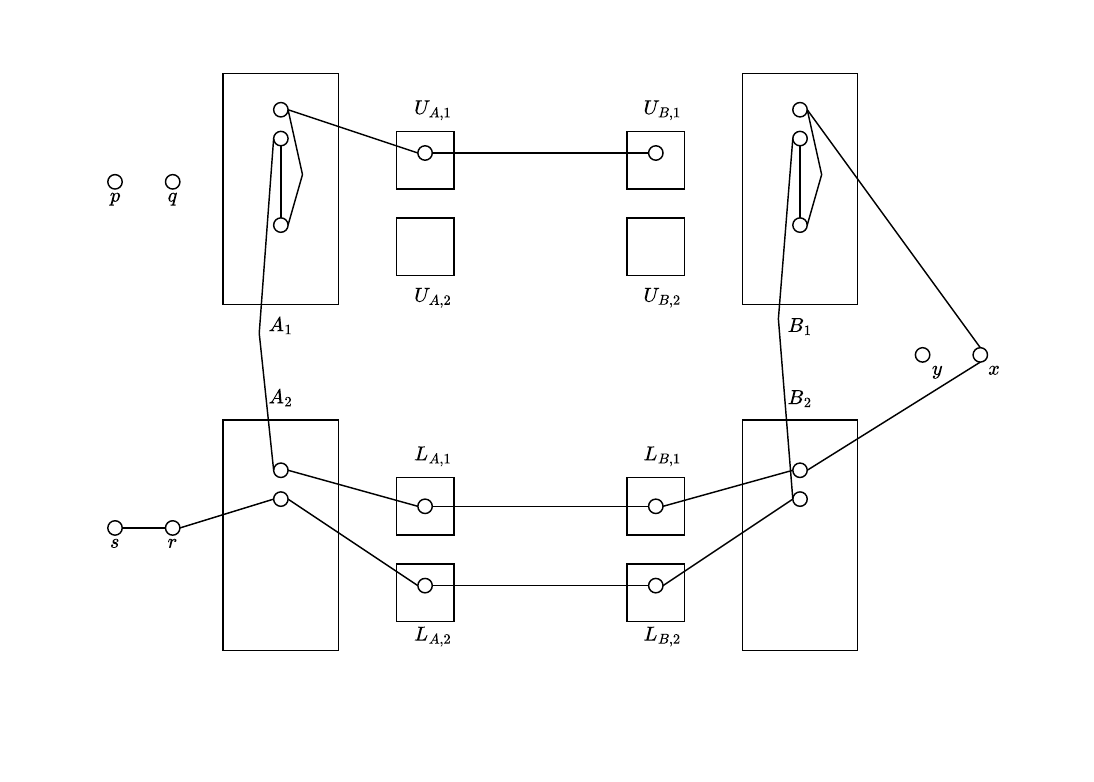}
    \end{center} 
\caption{An example of $P_{14}$ and $P_{19}$ in the proof of Lemma~\ref{P22_lemma1}.}
\label{fig:graph-P14-P19}
\end{figure}

\subsubsection{Lower bounds of the certificate size}

We can use the constructed graphs in this section to obtain lower bounds of the certificate size in the local certification of $P_k$-freeness stated in Theorem~\ref{thm:local-certification-lower-bound}. We use the linear lower bound on the nondetermistic communication complexity of set-disjointness as in~\cite{censor2020approximate}. Take a graph $G_{x,y}$ in Section~\ref{appendix-P22}, for instance, and replace each edge between $U_A\cup L_A$ and $U_B\cup L_B$ by a path of length 4. This makes the constructed graph $P_{28}$-free iff $\mathrm{DISJ}(x,y)=1$. All nodes in $A_1\cup A_2\cup U_A\cup L_A\cup \{p,q,r,s\}$ and the additional nodes in the replaced paths are simulated by Alice, and other nodes are simulated by Bob. 
Suppose that there is a local certification with locality 2 and certificate size $s$. Alice and Bob simulate it using $O(\sqrt{n}\cdot s)$ bits of communication. Using the fact that the nondetermistic communication complexity of set-disjointness on input size $\Theta(n^2)$ is $\Omega(n^2)$~\cite{Kushilevitz_Nisan_1996}, we get $s = \tilde{\Omega}(n^{3/2})$. Observe that, to increase the locality by 1, we have to add 8 additional nodes. Now we get the following statement: Any $\mathsf{PLS}_{\ell}$ that certifies $P_k$-freeness requires $\tilde{\Omega}(n^{3/2})$-bit certificate for $k\geq 8\ell + 14$. This corresponds to the statement of Theorem~\ref{thm:local-certification-lower-bound} for $d = 2$.

\section{Ordered path detection and applications}\label{sec:ordered-path}
In this section, we address the problem of detecting an ordered $P_k$. In this problem, each node of the graph is assigned a color from $\{1,\ldots,k\}$, and the objective is to detect an induced path $\{(p_i,p_{i+1})\}_{i\in \{1,\ldots,k-1\}}$ on $k$ nodes $\{p_1,\ldots p_k\}$, where each node $p_i$ is colored with $i$. Our first result is that detecting an ordered $P_5$ is already challenging, unlike the case of $P_k$-freeness where we only know nontrivial lower bound for $k\geq 11$ as in Theorem~\ref{thm:11-path-lower-bound}.

\begin{theorem}
 Any randomized algorithm that solves ordered $P_k$ detection for $k\geq 5$ requires $\tilde{\Omega}(n)$ rounds in the $\mathsf{CONGEST}$ model.
\end{theorem}
The proof of this theorem leverages a reduction from two-party communication complexity, with a construction that is notably simpler compared to the proof of Theorem~\ref{thm:11-path-lower-bound}. The detailed proof is provided in Appendix~\ref{appendix:ordered-P5}.

Our next insight regarding ordered path detection is the connection between proving non-trivial lower bounds for ordered $P_3$ detection and a long-standing open problem in circuit complexity.\footnote{This kind of connection is already known for several subgraphs such as triangles and $6$-cycles (see Section~\ref{appendix:related work}).} Moreover, this result has implications for the detection of induced $C_4$ in the quantum $\mathsf{CONGEST}$ model. 
Specifically, the following problem has remained unresolved for decades:

\begin{openproblem}\label{open-problem}
    Construct an explicit family of boolean functions $f_n:\{0,1\}^n\rightarrow \{0,1\}$ such that there exist constants $\delta_1,\delta_2 > 0$ such that any family of circuits $\{C_n\}_{n\in \mathbb{N}}$, where each circuit $C_n$ is depth $O(n^{\delta_1})$ and consists of $O(n^{1+\delta_2})$ gates with constant fan-in and fan-out, cannot compute $\{f_n\}_{n\in \mathbb{N}}$.
\end{openproblem}
The following theorem formalizes this relationship:

\begin{theorem}[The formal version of Theorem~\ref{thm:circuit_complexity_barrier_3_path}]\label{thm:ordered-3-path-circuit}
Let $\varepsilon>0$ be a constant. Then proving $\Omega(n^{\varepsilon})$ lower bound for ordered $P_3$ detection in the $\mathsf{CONGEST}$ model or $\Omega(n^{1/2+\varepsilon})$ lower bound for induced $C_4$ detection in the quantum $\mathsf{CONGEST}$ model solves Open Problem~\ref{open-problem}.
\end{theorem}

\begin{remark}
    One might think that ordered $P_3$ detection can be solved in $O(1)$ rounds as the case of $P_3$-freeness. However, it is easily shown that the round complexity of ordered $P_3$ detection in the $\mathsf{CONGEST}$ model is actually at least the round complexity of triangle detection in the $\mathsf{congested\;clique}$ model: Let $G'$ be the graph resulting by deleting all edges in the input graph $G$ between a node colored by 1 and a node colored by 3, and replacing non-edges between them by edges. Then any $\mathsf{CONGEST}$ algorithm for ordered $P_3$ detection is simulated in the $\mathsf{congested\;clique}$ model with the input graph $G'$ where each ordered $P_3$ in $G$ corresponds to a multicolored triangle in $G'$. Note that multicolored triangle detection can be reduced to (the standard) triangle detection in this model~\cite{nikabadi2022beyond}. 
\end{remark}

\begin{proof}[Proof sketch of Theorem~\ref{thm:ordered-3-path-circuit}]
    The first part of the statement builds on the known hardness results for triangle detection (and counting) established by~\cite{eden2022sublinear}. We demonstrate that ordered $P_3$ detection is at least as difficult as certain variants of triangle counting or ordered $P_3$ counting in high-conductance graphs. The full proof is deferred to Section~\ref{appendix:circuit-barrier}.

    For the second part, we use the framework of distributed quantum search~\cite{le2018sublinear}. Consider a function $f:X\rightarrow \{0,1\}$ defined on some finite set $X$, and assume a $T$-round $\mathsf{CONGEST}$ algorithm allows a fixed node $u$ to output $f(x)$ with constant probability, given input $x\in X$. In the quantum $\mathsf{CONGEST}$ model, we can find an element $x\in X$ such that $f(x)=1$ (if such an element exists) with constant probability in $\tilde{O}(\sqrt{|X|} \cdot T)$ rounds.

    Now, assume that ordered $P_3$ detection can be solved in $T$ rounds in the $\mathsf{CONGEST}$ model. We can show that induced $C_4$ detection can be solved in $\tilde{O}(\sqrt{n} \cdot T)$ rounds in the quantum $\mathsf{CONGEST}$ model. The statement then follows from the first part of the theorem, as an $\Omega(n^{1/2+\varepsilon})$ lower bound for induced $C_4$ detection implies an $\Omega(n^{\varepsilon})$ lower bound for ordered $P_3$ detection.

    Specifically, let $u$ be an arbitrary node, and define $M(u) = \{v: \text{dist}(u,v) = 2\}$. Consider the subgraph $G[N(u) \cup M(u)]$ induced by $N(u)$ and $M(u)$. Nodes in $M(u)$ are colored with color 2, while nodes in $N(u)$ randomly choose their color from 1 and 3. Note that $G[N(u) \cup M(u)]$ contains an ordered $P_3$ $(p_1,p_2,p_3)$ if $G$ contains an induced $C_4$ $(u, p_1, p_2, p_3)$. Since $p_1$ and $p_3$ randomly choose their colors, the probability that $(p_1,p_2,p_3)$ is properly colored is constant. Therefore, in $O(T)$ rounds, we can determine if $u$ belongs to an induced $C_4$. By using distributed quantum search for a function $f: V \rightarrow \{0,1\}$, which outputs $f(u) = 1$ if and only if $u$ belongs to an induced $C_4$, we can solve induced $C_4$ detection in $\tilde{O}(\sqrt{n} \cdot T)$ rounds.
\end{proof}

\subsection{Full Proof of Theorem~\ref{thm:ordered-3-path-circuit}}\label{appendix:circuit-barrier}

Now we prove the first part of Theorem~\ref{thm:ordered-3-path-circuit}.


\subsubsection{Reduction to a decision problem on high-conductance graphs}

We show a reduction from ordered $P_3$ detection on general graphs to a certain decision problem on high conductance graphs.

For a graph $G=(V,E)$, we say that a partition of edges $E= E_m\cup E_r$ is an $(\varepsilon,\phi)$-expander decomposition if the following conditions hold
\begin{enumerate}
    \item Let $V=V_1\cup \cdots \cup V_p$ be a set of connected components induced by the edge set $E_m$. Let $G\{V_i\}$ be the graph resulting from adding $\mathrm{deg}_V(u) - \mathrm{deg}_{V_i}(u)$ self loops to each node $u$ in $G[V_i]$. The conductance of $G\{V_i\}$ satisfies $\Phi(G\{V_i\})\geq \phi$ for each $i$. 
    \item $|E_r|<\varepsilon |E|$.
\end{enumerate}

\begin{theorem}[\cite{chang2021near}]
For any $\varepsilon\in (0,1)$ and a positive integer $k$, an $(\varepsilon,\phi)$-expander decomposition with $\phi = (\varepsilon/\mathrm{polylog}(n))^{20\cdot 3^k}$ can be constructed in 
$
O\left(
(\varepsilon m)^{1/k}\cdot \left( \frac{\mathrm{polylog}(n)}{\varepsilon}\right)^{20\cdot 3^k}
\right)
$ rounds with high probability.
\end{theorem}

We next define the following intermediate problem which we call \textsf{edge-colored triangle counting}.

\begin{problem}[\textsf{edge-colored triangle counting}] \label{problem:edge-colored-triangle}
     Let $G=(V^{in},E^{in})$ is an $n$-node graph where each node is colored by 1 or 2 or 3. For each edge $e$ let $c(e)\in \{0,1\}^K$ for some $K=O(\log n)$ be the color known to its endpoints. Let $z$ be a function (known to all nodes) that maps a multicolored triangle $T$ (a triangle whose three nodes are colored by 1,2,3) in the graph to an integer $z(T)$ from $\{0,1,\ldots, \mathrm{poly}(n)\}$. 
    The goal of this problem is to compute the sum of $z(T)$ for all multicolored triangles $T$ in $G$.
\end{problem}

Here we set the function $z$ as follows. Let $T=(e_1,e_2,e_3)\in E^{in}\times E^{in}\times E^{in}$ be a multicolored triangle (a triangle whose three nodes are colored by 1,2,3) such that the endpoints of $e_1$ are colored 1 and 3. Let $z(T)$ be a number determined as follows: If all edges are colored by $0^K$ then $z(T)=0$. Let $i_p\in [K]$ be the minimum integer such that the $i_p$-th bit in $c(e_p)$ is 1 for $p\in \{1,2,3\}$. If $i_1 \leq i_2 \leq i_3$ does not hold then $z(T)=0$. If $i_1 = i_2$, then $z(T)=0$. Otherwise, $z(T)$ is the number of 1's in the $j$-th bit of $c(e_1)$ for $j\in \{i_1+1,\ldots, i_2\}$.

Let $\mathcal{A}_1$ be an algorithm for \textsf{edge-colored triangle counting} in graphs with conductance $\phi$, and $T(n,\phi)$ be its round complexity. Let $\mathcal{A}_2$ be an algorithm for ordered $3$-path counting in graphs with conductance $\phi$, and $T'(n,\phi)$ be its round complexity. We consider the following protocol.

\begin{enumerate}
\item We first apply $(\varepsilon,1/\mathrm{polylog}(n))$-expander decomposition in $K = O(\log n)$ times recursively (i.e., the $(i+1)$-th expander decomposition is excuted on the edges between two clusters in the $i$-th expander decomposition). Let $E^{(i)}_m \cup E^{(i)}_r$ be the partition of edges by the $i$-th expander decomposition where $E^{(i)}_m$ is the set of cluster edges. Then, we can divide the edge set as
$E=E^{(1)}_m\cup E^{(2)}_m \cup \cdots \cup E^{(K)}_m \cup E^{(K)}_r$, and assume that $|E^{(K)}_r| = \tilde{O}(1)$. For each edge $e=(u,v)$, we define its color $c(e)$ as follows: For each $i\in [K]$, if $u$ and $v$ are in the same cluster in the $i$-th expander decomposition, then the $i$-th bit of $c(e)$ is 1, otherwise 0.
\item For each $i\in [K]$, each cluster of the graph $G^{(i)}=(V,E^{(i)}_m \cup E^{(i)}_r)$ runs $\mathcal{A}_1$ in parallel, where the input edges for each node is all edges incident to them.  After that, the leader node learns $E^{(K)}_r$ in $|E^{(K)}_r| = \tilde{O}(1)$ rounds and locally compute the solution of \textsf{edge-colored triangle counting} in the graph $(V,E^{(K)}_r)$. Let $A_1$ be the solution of \textsf{edge-colored triangle counting} in the original graph $G$.
\item For each $i\in [K]$, each cluster of the graph $G^{(i)}=(V,E^{(i)}_m \cup E^{(i)}_r)$ runs $\mathcal{A}_2$ in parallel, where the input edges for each node is all edges incident to them.
After that, the leader node learns $E^{(K)}_r$ in $|E^{(K)}_r| = \tilde{O}(1)$ rounds and locally compute the number of ordered $P_3$'s in the graph $(V,E^{(K)}_r)$. Let $A_2$ be the sum of all counted $P_3$'s in the whole procedure.
\end{enumerate}

\begin{lemma}
    $A_1 = A_2$ if and only if $G$ does not contain an ordered $3$-path.
\end{lemma}

\begin{proof}
     Let $T=(e_1,e_2,e_3)$ be a set of three edges that form a multicolored triangle such that two endpoints of $e_1$ is colored 1 and 3, and $z(T)\neq 0$. From the definition of edge colors, for $i_1,i_2,i_3\in [K]$ in Problem~\ref{problem:edge-colored-triangle}, the edge $e_1$ is removed from the graph after the $i_1$-th expander decomposition and both edges $e_2$ and $e_3$ remains before the $i_2$-th expander decomposition. Recall that $z(T)$ is the number of 1's in the $j$-th bit of $c(e_1)$ for $j\in \{i_1+1,\ldots, i_2\}$, and the $j$-th bit of $c(e_1)$ is 1 iff both endpoints of $e_1$ are in the same cluster in the $j$-th expander decomposition. For such $j$, two edges $e_2,e_3$ remain in the graph $G^{(j)}=(V,E^{(j)}_m \cup E^{(j)}_r)$. Therefore $T$ is counted as an ordered $P_3$ in $G^{(j)}$ if and only if $j\in \{i_1+1,\ldots, i_2\}$. This indicates that $A_1 \leq A_2$. In particular, $A_1 = A_2$ iff $G$ contains no ordered $3$-path.
\end{proof}
We thus obtain the following.
\begin{lemma}\label{lem:ordered-3-path-reduction}
We can solve ordered $P_3$ detection in $\tilde{O}(T(n,\mathrm{polylog}(n)) + T'(n,\mathrm{polylog}(n)))$ rounds.
\end{lemma}

\subsubsection{Constructing explicit functions}

Consider an $n$-node $m$-edge connected graph $G =(V,E)$ with conductance at least $1/\mathrm{polylog}(n)$ where each node is colored by an integer from $\{1,2,3\}$ and each edge is colored by a $K=O(\log m)$-bit string. This graph is encoded by $N: = (K+2\log m + 4)m = O(m\log m)$ bits (where $K$ represents the color of each edge, $2\log m + 4$ represents the identifiers and the colors of two endpoints of each edge). 
We then define a boolean function $f:\{0,1\}^{N}\times \{0,1\}^{O(\log n)}\times \{0,1\} \rightarrow \{0,1\}$ as follows: Let $(x,i,j)\in \{0,1\}^{N}\times \{0,1\}^{O(\log n)}\times \{0,1\}$ is the input of $f$ where $x$ is considered as the representation of the input graph. The output is the $i$-th bit of the answer of \textsf{edge-colored triangle counting}\footnote{Since each $z(T)$ is represented by $O(\log n)$ bits and the number of triangles in $G$ is at most $O(n^3)$, $A_1$ is represented by $O(\log n)$ bits.} if $j=0$, otherwise the output is the $i$-th bit of the number of ordered $3$-paths in the input graph.
Assume that there exists a family of circuits with depth $O(N^{\alpha/5})$ and size $\mathsf{S}(f)$ that computes $f$. 
Now we use Lemma~\ref{lem:circuit-simulation} to show that proving polynomial lower bounds for \textsf{edge-colored triangle counting} or \textsf{ordered $3$-path counting} implies a super-linear lower bound on $\mathsf{S}(f)$.
\begin{lemma}[Lemma 26 of \cite{eden2022sublinear}]\label{lem:circuit-simulation}
    Let $U$ be a graph $U=(V,E)$ with $n$ nodes, $n^{1+\delta}$ edges, and mixing time $\tau$. Assume that a boolean function $f:\{0,1\}^{cn^{1+\delta}\log n}\rightarrow \{0,1\}$ is computed by a circuit $\mathsf{C}$ of depth $R$, comprising of gates with constant fan-in and fan-out, and at most $O(c\cdot s \cdot n^{1+\delta}\log n)$ wires for $s\leq n$. Then for any input partition that assigns to each node $u\in U$ no more than $c\cdot\mathrm{deg}(u)$ input wires, there is an $O(R\cdot c \cdot s \cdot \tau \cdot 2^{O(\sqrt{\log n \log \log n})})$-round protocol in the $\mathsf{CONGEST}$ model on $U$ that computes $f$ under the input partition.
\end{lemma}

\begin{lemma}\label{lem:circuit_lower_bound}
    If \textsf{edge-colored triangle counting} on $n$-node graphs with mixing time $\mathrm{polylog}(n)$ requires $\Omega(n^{\alpha})$ rounds, or \textsf{ordered $P_3$ counting} on $n$-node graphs with mixing time $\mathrm{polylog}(n)$ requires $\Omega(n^{\alpha})$ rounds, then $\mathsf{S}(f) = \Omega(N^{1+\alpha/4})$.
\end{lemma}

\begin{proof}
    Assume that there exists a family of circuits with depth $O(N^{\alpha/5})$ and size $\mathsf{S}(f) = O(N^{1+\alpha/4 + o(1)})$ that solves the $\Theta(N)$-bit boolean functions $f$ for all $N$. Then by Lemma~\ref{lem:circuit-simulation} with $s = N^{\alpha/4 + o(1)}$ and $R=N^{\alpha/5}$, there exists a $\mathsf{CONGEST}$ algorithm that computes $f$ within \begin{align*}
        O(R\cdot c \cdot s \cdot \tau \cdot 2^{O(\sqrt{\log n \log \log n})}\cdot \log n) = O(n^{0.9\alpha + o(1)})
    \end{align*} 
    rounds. Using this algorithm for $O(\log n)$ different inputs we can solve \textsf{edge-colored triangle counting} and \textsf{ordered $P_3$ counting} in $\tilde{O}(n^{0.9\alpha + o(1)})$ rounds, which is a contradiction.
\end{proof}

Combining Lemma~\ref{lem:ordered-3-path-reduction} and Lemma~\ref{lem:circuit_lower_bound} we get the first part of Theorem~\ref{thm:ordered-3-path-circuit}. In Appendix~\ref{appendix:induced-C4}, we complement this barrier by showing $O(n^{3/4})$ round quantum algorithm for induced $C_4$ detection as in Theorem~\ref{thm:induced-4-cycle}.

\bibliography{mybibliography}

\begin{appendix}

\section{Lower bound for ordered $P_5$ detection}\label{appendix:ordered-P5}

The proof uses the same technique for the case of $P_{k}$-freeness: the reduction from the two-party communication complexity. Thus what we need to do is the graph construction.

We construct the graph $G_{x,y}$ as follows.
The node set $V$ of $G_{x,y}$ is the union of two subsets $V_A$ and $V_B$.
$V_A$ consists of three node sets $A_1,A_2,A_5$ of size $n$. $V_B$ consists of two node sets $B_3,B_4$ of size $n$.
Each set is explicitly described as $A_i = \{a_{i,j}|j\in [n]\}$ for $i\in \{1,2,5\}$ and $B_i = \{b_{i,j}|j\in [n]\}$ for $i\in \{3,4\}$.
The following edges are contained for all $G_{x,y}$, regardless of the value of $x$ and $y$.
\begin{itemize}
    \item $(a_{5,j},b_{4,j})$, $(a_{2,j},b_{3,j})$ for all $j\in [n]$;
    \item $(a_{1,j},a_{5,k})$ for all distinct $j,k\in [n]$.
\end{itemize}
Furthermore, we add the following edges depending on the value of $x$ and $y$. 
For each $j_1,j_2\in [n]$, let $k = j_1 + j_2\cdot (n-1)$ and 
\begin{itemize}
    \item if $x_k = 1$, we add the edge $(a_{1,j_1},a_{2,j_2})$;
    \item if $y_k = 1$, we add the edge $(b_{4,j_1},b_{3,j_2})$.
\end{itemize}
Finally, the color of nodes in $A_i$ is defined to be $i$ for $i\in \{1,2,5\}$, and the color of nodes in $B_i$ is defined to be $i$ for $i\in \{3,4\}$.
This finishes the construction of the graph $G_{x,y}$. We then prove the following lemma. The proof of Theorem~\ref{thm:ordered-5-path} is exactly the same as how Theorem~\ref{thm:11-path-lower-bound} was proved from Lemma~\ref{lem:11-path-graph-construction} (using Lemma~\ref{lem:disjointness_and_5_path} instead). 

\begin{lemma}\label{lem:disjointness_and_5_path}
    $G_{x,y}$ contains a copy of ordered $5$-path if and only if $\mathsf{DISJ}_{n^2}(x,y)=0$.
\end{lemma}
\begin{proof}
    Any possible copy is required to use edges of the form $(a_{5,j_1},b_{4,j_1})$ and $(a_{2,j_2},b_{3,j_2})$. Thus it also contains edges of the form  $(a_{1,j_1},b_{2,j_2})$ for the same indices $j_1,j_2$. This is because since all nodes of color 1 except $a_{1,j_1}$ are incident to $a_{5,j_1}$ and they cannot be a part of the desired path. Finally, it is straightforward to confirm the nodes $a_{1,j_1},a_{2,j_2},b_{3,j_2},b_{4,j_1},a_{5,j_1}$ form the desired path if and only if $x_k=y_k=1$ for $k = j_1 + j_2\cdot (n-1)$.
\end{proof}
Since we have $O(n)$ edges between $V_A$ and $V_B$, we get $\tilde{\Omega}(n)$ lower bound.

\section{Quantum algorithm for induced $C_4$ detection}\label{appendix:induced-C4}

Our algorithm for induced $C_4$ detection is governed by a threshold $\delta$. We say that a vertex is light iff its degree is at most $\delta$. An induced $C_4$ is light if at least one of its vertices is light, otherwise the cycle is heavy. 
\paragraph*{Detecting light cycles}
In the algorithm each light node first broadcasts its list of neighbors in $O(\delta)$ rounds. Each node $v$ then constructs, for each $(v,w)\in E$, a list $S(v,w)$ as follows:
\[
S(v,w)=\{
    w'\in V \mid \exists u\in V,\mathrm{deg}(u)\leq \delta, w'\in N(u)\backslash N(v)
\}
\]
Note that $v$ can construct $S(v,w)$ locally. For each edge $(v,w)$, $v$ and $w$ search the intersection of $S(v,w)$ and $N(w)$ using distributed quantum search used in the proof of Theorem~\ref{thm:circuit_complexity_barrier_3_path}. This can be done $$O(\sqrt{\mathrm{min}\{|S(v,w)|,|N(w)|\}})=O(\sqrt{n})$$ rounds. Therefore detecting light cycles requires $O(\delta + \sqrt{n})$ rounds.

\paragraph*{Detecting heavy cycles}
Let us now consider a fixed induced $C_4$: $(u,v),(v,w),(w,x),(x,u)\in E$ such that all nodes $u,v,w,x\in V$ are heavy. Consider that each node selects $n/\delta$ edges from the set of its incident edges uniformly at random, and broadcasts them. Then, $u$ constructs for each edge $(u,v)\in E$
\[
S(v)=\{
    x \in V \mid \exists w\notin N(u) \text{ such that }(v,w)\in E \text{ and } (w,x)\in E
\}
\]
Define the function $f_v^u:V\rightarrow \{0,1\}$ where $f_v^u(x)=1$ iff $x\in S(v)\backslash N(v)$. Using Lemma~\ref{lem:distributed-quantum-search}, we can determine if there exists $x$ for which $f_v^u(x)=1$ using the edge $(u,v)$ in $O(\sqrt{n})$ rounds. Furthermore, this procedure is done in parallel for all $\{f_v^u\}_{(u,v)\in E}$.

Since both the edges $(v,w)$, $(w,x)$ are sent to $u$ with porobability at least $1/\delta^2$, repeating this procedure $O(\delta^2)$ times we can detect heavy cycles with constant probability. For this kind of randomized algorithms, we can apply a technique called amplitude amplification, which is a generalization of quantum search.

\begin{lemma}[\cite{fraigniaud2024even}]\label{lem:amplitude-amplification}
    Let $\mathcal{P}$ be a Boolean predicate on
graphs. Assume that there exists a (randomized or quantum) distributed algorithm $\mathcal{A}$ that decides $\mathcal{P}$ with one-sided success probability $\varepsilon$ on input graph $G$, i.e.,
\begin{itemize}
    \item If $G$ satisfies $\mathcal{P}$, then, with probability 1, $\mathcal{A}$ accepts at all nodes;
    \item If $G$ does not satisfy $\mathcal{P}$, then, with probability at least $\varepsilon$, $\mathcal{A}$ rejects in at least one node.
\end{itemize}
Assume further that $\mathcal{A}$ has round-complexity $T (n, D)$ for $n$-node graphs of diameter at most $D$.
Then, for any $\delta > 0$, there exists a quantum distributed algorithm $\mathcal{B}$ that decides $\mathcal{P}$ with one-sided
error probability $\delta$, and round-complexity $\mathrm{polylog}(1/\delta) \cdot \frac{1}{\sqrt{\varepsilon}} (D + T (n, D))$.
\end{lemma}
Note that the diameter reduction technique~\cite{eden2022sublinear} allows us to assume the diameter is $O(\log n)$.
Using Lemma~\ref{lem:amplitude-amplification} with $T=O\left(\frac{n}{\delta} + \sqrt{n}\right)$ and $\varepsilon=O\left(\frac{1}{\delta^2}\right)$, we can find an heavy cycle (if exists) with constant probability within the round complexity $O\left(\left(\frac{n}{\delta} + \sqrt{n}\right)\delta\right)$.

We now consider the following protocol in which both the edges $(v,w)$, $(w,x)$ are sent to $u$ with probability at least $1/\delta$, instead of $1/\delta^2$. 
In the algorithm each node $u$ first selects an integer $\mathsf{ID}(u)$ from $[\delta]$ uniformly at random and broadcasts it to all its neighbors. Then $u$ divide $N(u)$ into $\delta$ subsets $N_i(u)$ for $i\in[\delta]$ where
$$
N_i(u) = \{
v\in N(u)\mid \mathsf{ID}(v) = i.
\}
$$
Note that $|N_i(u)|=O\left(\frac{|N(u)|}{\delta}\right) = O\left(\frac{n}{\delta}\right)$ with high probability. We thus condition on this event. Each node $u$ then selects $i\in [\delta]$ uniformly at random and asks each of its neighbors $v$ to send $N_i(v)$. In this procedure the edges $(v,w)$, $(w,x)$ are sent to $u$ with probability at least $1/\delta$. This leads to the round complexity $O\left(\left(\frac{n}{\delta} + \sqrt{n}\right)\sqrt{\delta}\right)$. Taking $\delta=\sqrt{n}$, we get the $O(n^{3/4})$-round algorithm for detecting induced $C_4$, as stated in Theorem~\ref{thm:induced-4-cycle}.

\end{appendix}

\end{document}